\documentclass{IEEEtran}
\def\BibTeX{{\rm B\kern-.05em{\sc i\kern-.025em b}\kern-.08em
    T\kern-.1667em\lower.7ex\hbox{E}\kern-.125emX}}

\usepackage{textcomp}
\usepackage{cite}
\usepackage{amsmath,amssymb,amsfonts}
\usepackage{algorithmic}
\usepackage{graphicx}
\usepackage{multirow}
\usepackage{array}
\usepackage{subcaption}
\usepackage{xcolor}
\usepackage{makecell} 
\usepackage{booktabs}
\usepackage{tabularx}
\usepackage[hidelinks]{hyperref}
\hypersetup{
    citecolor = black,
    }

\newtheorem{remark}{Remark}
\newtheorem{theorem}{Theorem}

\newtheorem{assumption}{Assumption}

\newcommand{\norm}[1]{\left\Vert#1\right\Vert}
\newcommand{\bea}{\begin{eqnarray}}
\newcommand{\eea}{\end{eqnarray}}
\newcommand{\beas}{\begin{eqnarray*}}
	\newcommand{\eeas}{\end{eqnarray*}}

\graphicspath{{./}}

\title{Input-Constrained Active Disturbance Rejection Control for AUVs
}

\author{Minh-Quan Le, Ying Tan, Jinghe Yang, Jason Monty, and Ye Pu%
{\thanks{Authors are with the School of Electrical, Mechanical and Infrastructure Engineering, University of Melbourne, Parkville, VIC 3010, Australia.}%
\thanks{Corresponding author: Minh-Quan Le (e-mail: \texttt{quan.le@student.unimelb.edu.au}).}%
}}

\begin{document}

\maketitle

%%%%%%%%%%%%%%%%%%%%%%%%%%%%%%%%%%%%%%%%%%%%%%%%%%%%%%%%%%%%%%%%%%%%%%%%%%%%%%%%
\begin{abstract}
This paper addresses the control of autonomous underwater vehicles (AUVs) operating in dynamic environments, where wave-induced forces, modelling errors, and actuator limitations make reliable control challenging. Active disturbance rejection control (ADRC), together with an extended state observer (ESO), estimates lumped uncertainties online, simplifying control design using nominal models. However, existing stability analyses typically assume bounded control inputs and derivatives, which are difficult to guarantee under unknown disturbances and measurement noise. To bridge this theory--practice gap, this paper introduces an ADRC framework with active input saturation and post-saturation low-pass filtering. The saturation constrains the control magnitude, while the filter regularizes the saturated signal to ensure bounded input-rate variations. By explicitly enforcing these conditions, the proposed system guarantees semi-global practical asymptotic stability. High-fidelity simulations and experiments on an underwater robotic platform demonstrate robust control performance under severe environmental disturbances and actuator constraints. A supplementary video is available at {\em \url{https://youtu.be/AEFlZ69Yg00}}.
\end{abstract}

\begin{IEEEkeywords}
Active disturbance rejection control, Extended state observer, Autonomous underwater vehicles, Input saturation and rate-limited control.
\end{IEEEkeywords}

% This paper addresses the control of autonomous underwater vehicles (AUVs) operating in dynamic environments. Strong uncertainties, including wave-induced forces, modelling errors, and actuator limitations, make reliable control particularly challenging. Active disturbance rejection control (ADRC), together with an extended state observer (ESO), is widely used to estimate the lumped uncertainties online to simplify the control to the nominal models. However, existing stability analyses typically assume bounded control inputs and derivatives, an assumption hard to be guaranteed in practice under unknown disturbances and measurement noise. To bridge this theory-to-implementation gap, this paper introduces an ADRC framework with active input saturation and post-saturation low-pass filtering. Active saturation constrains the control magnitude, while the low-pass filter regularizes the saturated control signal to ensure bounded input-rate variations. By explicitly enforcing these boundedness conditions, the system theoretically guarantees semi-global practical asymptotic stability. The proposed approach is validated through high-fidelity simulations and experiments on an underwater robotic platform, demonstrating robust control performance under severe environmental disturbances and actuator constraints. (A supplementary video summarizing this work is available at https://youtu.be/AEFlZ69Yg00)

%%%%%%%%%%%%%%%%%%%%%%%%%%%%%%%%%%%%%%%%%%%%%%%%%%%%%%%%%%%%%%%%%%%%%%%%%%%%

\section{Introduction}

Autonomous underwater vehicles (AUVs) are increasingly deployed for ecosystem monitoring, offshore infrastructure inspection, environmental surveying, and search-and-rescue operations \cite{Bashir_Review_2023}. Reliable stabilization is essential because AUVs operate under complex hydrodynamics, ocean currents, waves, and modelling uncertainties \cite{fossen,Bashir_Review_2023}. Practical actuator limitations, including input magnitude and rate constraints and dead-zone nonlinearities, further complicate control design \cite{blue,t200}. Developing robust strategies that reject environmental disturbances while respecting these limitations therefore remains an important research challenge.

Classical proportional--integral--derivative (PID) controllers remain widely used because of their simplicity and ease of implementation, but their performance often deteriorates in nonlinear and disturbance-rich underwater environments \cite{pid2,pid3,pid4,pid5}.
Adaptive methods, including model reference adaptive control (MRAC) and $\mathcal{L}1$ adaptive control, compensate for parametric uncertainties through online adaptation \cite{nullspace,adap1,adap2,adap3}. 
Model predictive control (MPC) explicitly incorporates system constraints through receding-horizon optimization and has demonstrated strong performance for underwater vehicles \cite{mpc1,mpc3,mpc4,mpc5,mpc6}. 
Sliding mode control (SMC) and $H_\infty$ control provide robustness against disturbances and modelling uncertainties \cite{blue, blue2, smc1,smc2,h-inf1,h-inf2,h-inf3}. 
However, adaptive control often requires restrictive uncertainty parameterizations, MPC incurs substantial online computational costs, and SMC may cause undesirable chattering or aggressive control actions \cite{smc1,smc2,mpc1,adap3}. 
More recently, data-driven approaches, including supervised and reinforcement learning, have been investigated to model complex underwater dynamics and improve control performance \cite{learn,learn1,learn2,learn3}. 
Despite their adaptability, establishing rigorous stability and robustness guarantees under severe disturbances and actuator constraints remains challenging.

Active disturbance rejection control (ADRC) offers an attractive alternative by using an extended state observer (ESO) to estimate lumped uncertainties online and compensate for them through the control law \cite{Han2009ADRC,Guo2015ADRCTheory,adrc1,adrc2}. 
Its reduced dependence on accurate models and relatively low computational complexity make it suitable for embedded systems \cite{Han2009ADRC,Guo2015ADRCTheory}. 
ADRC has therefore gained increasing attention in underwater robotics for vehicle stabilization, heading control, depth regulation, and hybrid MPC--ADRC architectures \cite{adrc_linear,adrc_heading,adrc_depth,adrc_mpc}. 
Nevertheless, a significant gap remains between ADRC theory and practical implementation.
Existing stability analyses typically assume bounded lumped uncertainties, control inputs, and their time derivatives \cite{Han2009ADRC,Guo2015ADRCTheory,adrc1}. 
Although these assumptions facilitate ESO convergence and closed-loop stability analysis, they are difficult to guarantee in applications, where environmental disturbances and measurement noise can produce large control commands and rapid input changes.
Thus, satisfying the input-related assumptions of existing ADRC theory remains challenging for actuator-limited underwater vehicles.

To bridge this theory--practice gap, this paper develops an input-constrained active disturbance rejection control (IC-ADRC) framework combining active input saturation and post-saturation low-pass filtering.
The saturation bounds the control magnitude according to actuator capabilities, while the filter limits input-rate variations. 
Singular perturbation analysis establishes the required time-scale separation among the plant, filter, and ESO dynamics, leading to a semi-global practical stability guarantee.

The main contributions are:
\begin{itemize}
\item An IC-ADRC framework that combines active input saturation and post-saturation filtering to enforce bounded and smooth control inputs for AUVs.

\item A singular perturbation-based design methodology that establishes the required time-scale separation among the plant, filter, and ESO dynamics.

\item A theoretical proof of semi-global practical stability for the resulting closed-loop system.

\end{itemize}

The proposed framework is validated through high-fidelity simulations and experiments on an underwater robotic platform under environmental disturbances and actuator constraints.

%%%%%%%%%%%%%%%%%%%%%%%%%%%%%%%%%%%%%%%%%%%%%%%%%%%%%%
\section{Preliminaries and Problem Formulation}

\label{sec:pre}

This section presents the notation, the dynamic modeling of AUVs, and the problem formulation addressed in this work.

Let $\mathbb{R}$ and $\mathbb{R}^n$ denote the sets of real numbers and $n$-dimensional Euclidean space, respectively. 
For a vector $\boldsymbol{x} \in \mathbb{R}^n$, its Euclidean norm is defined as $\norm{\boldsymbol{x}} = \sqrt{\boldsymbol{x}^\top \boldsymbol{x}}$.

The entrywise saturation $\mathrm{sat}(\boldsymbol{u}, \boldsymbol{u}_{\max})$ for $\boldsymbol{u} \in \mathbb{R}^n$ and positive bound vector $\boldsymbol{u}_{\max} \in \mathbb{R}^n_{>0}$ is defined for each component $i = 1, \dots, n$ by:
\begin{equation} \label{eqn:vector_saturation_def}
    [\mathrm{sat}(\boldsymbol{u}, \boldsymbol{u}_{\max})]^i =  \begin{cases} \operatorname{sign}(u^i) u^i_{\max}, & \vert{}u^i\vert{} > u^i_{\max}, \\ u^i, & \vert{}u^i\vert{} \le u^i_{\max}, \end{cases}   
\end{equation}

\subsection{Vehicle Kinematics and Dynamics}
In this work, we consider the widely adopted second-order, six-degree-of-freedom (6-DOF) dynamic model of a fully actuated autonomous underwater vehicle (AUV), expressed in the North--East--Down (NED) inertial reference frame, as presented in \cite{fossen}. 

\begin{figure}[!htb]
    \centering
    \includegraphics[width=0.7\linewidth]{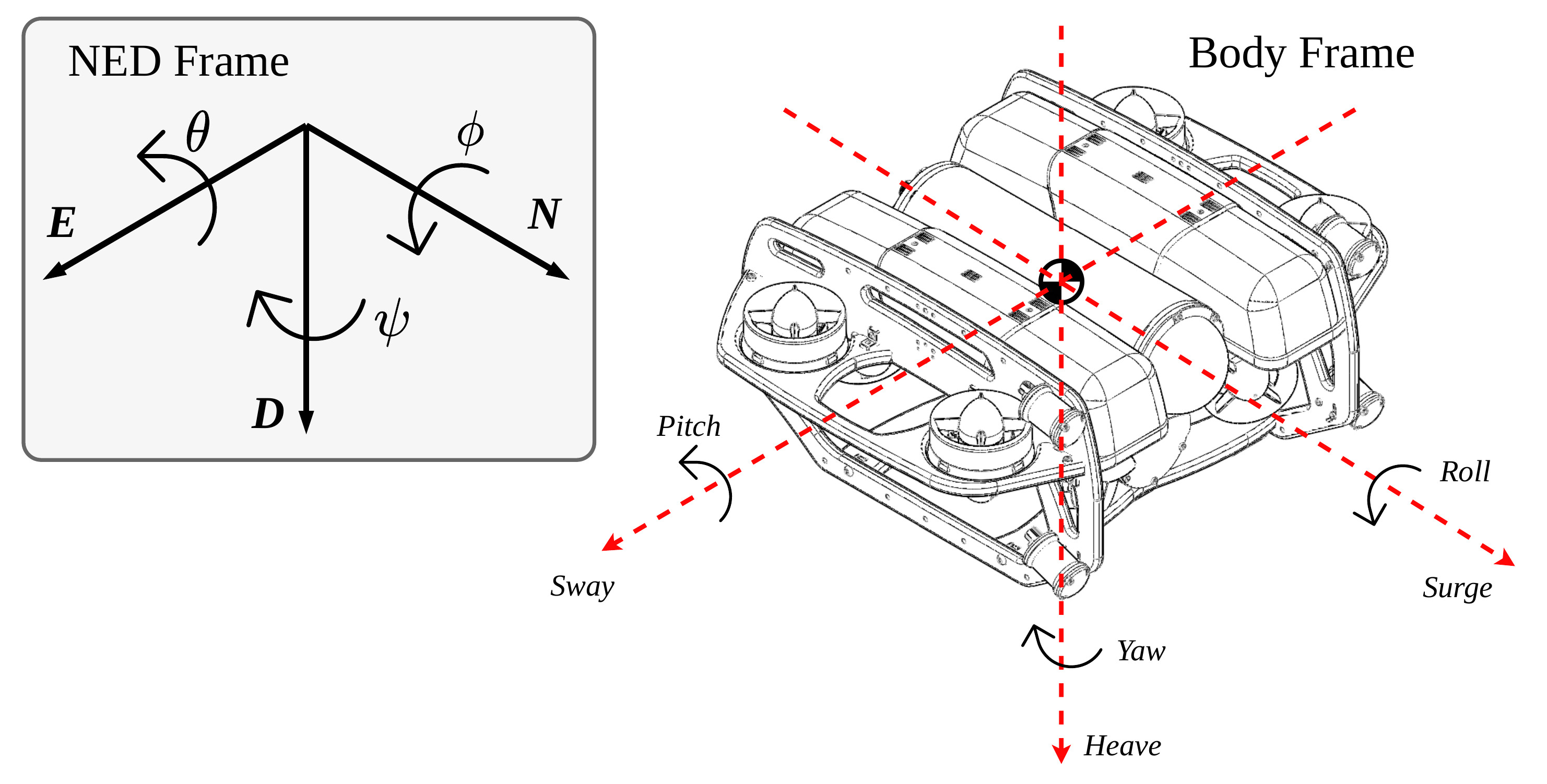}
    \caption{NED inertial and body-fixed frames of an AUV \cite{fossen}.}
    \label{fig:frames}
\end{figure}

The pose of the vehicle in the NED frame is denoted by
\[
\boldsymbol{\eta}
=
[x \ y \ z \ \phi \ \theta \ \psi]^\top \in \mathbb{R}^6 ,
\]
where $(x,y,z)$ are the position coordinates and
$(\phi,\theta,\psi)$ are the roll, pitch, and yaw angles.
To formulate the control problem directly in the task space, we define the generalized velocity of the vehicle strictly in the NED frame as $\boldsymbol{v}
= \dot{\boldsymbol{\eta}} \in \mathbb{R}^6$.

Expressed entirely in the NED frame, the 6-DOF equations of motion for the vehicle are given by:
\begin{align}
\dot{\boldsymbol{\eta}} = {\boldsymbol{v}}, \hspace{2.55in} \label{eqn:kinematics}\\
\boldsymbol{M}(\boldsymbol{\eta})\dot{\boldsymbol{v}}
+
\boldsymbol{C}_{Cor}(\cdot)\boldsymbol{v}
+
\boldsymbol{D}(\cdot)\boldsymbol{v}
+
\boldsymbol{g}(\boldsymbol{\eta})
=
\boldsymbol{\tau}
+
\boldsymbol{\tau}_{\mathrm{wave}},
\label{eqn:dynamics}
\end{align}
where
$\boldsymbol{M}(\cdot) \in \mathbb{R}^{6\times6}$ is the symmetric and positive definite inertia matrix. 
The notion of 
$\boldsymbol{C}_{Cor}(\boldsymbol{v},\boldsymbol{\eta})$ is the Coriolis–centripetal matrix,
$\boldsymbol{D}(\boldsymbol{v},\boldsymbol{\eta})$ models hydrodynamic damping,
$\boldsymbol{g}(\boldsymbol{\eta})$ represents hydrostatic forces,
$\boldsymbol{\tau}$ is the control input (thruster forces and moments),
and
$\boldsymbol{\tau}_{\mathrm{wave}}$ denotes environmental disturbances,
dominated here by wave-induced loads. 
Here, we denote $\boldsymbol{\tau} \in \mathbb{R}^6$ as the generalized control input expressed in the NED frame. 
Because the  thrusters are rigidly attached to the vehicle, allocating this input requires transforming it into the body-fixed frame via:
\begin{equation}
    \boldsymbol{\tau}^b = \boldsymbol{J}^\top(\boldsymbol{\eta}) \boldsymbol{\tau},
    \label{eqn:tau_body_transform}
\end{equation}
where $\boldsymbol{J}(\boldsymbol{\eta})$ is the standard 6-DOF transformation matrix mapping the body and NEW frames \cite{fossen}.
The system measures the vehicle pose as output
\begin{equation}\label{eqn:measurement}
\boldsymbol{y} = \boldsymbol{\eta}.
\end{equation}

\subsection{Problem Formulation}

To facilitate controller design, the nonlinear system \eqref{eqn:kinematics}, \eqref{eqn:dynamics} is linearized around the equilibrium point 
$(\boldsymbol{\eta},\boldsymbol{v}) = \boldsymbol{0}_{12}$. 
All unmodelled dynamics, nonlinearities, linearization errors, and external disturbances are grouped into a lumped uncertainty term 
$\boldsymbol{\sigma} \in \mathbb{R}^6$, which can be expressed as
\begin{equation}
\boldsymbol{\sigma} := \boldsymbol{\sigma}(\boldsymbol{\eta}, \boldsymbol{v}, \boldsymbol{\tau}, \boldsymbol{\tau}_{\mathrm{wave}}).
\label{eqn_lumped_uncertainties}
\end{equation}
The model functions $\boldsymbol{M}$,
$\boldsymbol{C}_{\mathrm{Cor}}$, $\boldsymbol{D}$, and $\boldsymbol{g}$
are continuously differentiable on their entire domains, and
$\boldsymbol{M}(\boldsymbol{\eta})$ is nonsingular for all
$\boldsymbol{\eta}\in\mathbb{R}^{6}$. Consequently, the lumped uncertainty
$\boldsymbol{\sigma}$ is continuously differentiable with respect to its
arguments.
Defining the state vector
\begin{equation}
    \boldsymbol{x}
    =
    \begin{bmatrix}
        \boldsymbol{\eta}^{\top} &
        \boldsymbol{v}^{\top}
    \end{bmatrix}^{\top}
    \in \mathbb{R}^{12}.
    \label{eqn:regulation_state}
\end{equation}
With $\boldsymbol{M}_n=\boldsymbol{M}(\boldsymbol{0})$, it has the following linearized dynamics:
\begin{align}\label{eqn:linear_system}
\dot{\boldsymbol{x}} &= \boldsymbol{A}_n \boldsymbol{x} + \boldsymbol{B}_n \left(\boldsymbol{\tau} + \boldsymbol{\sigma}\right),\\
\boldsymbol{y} &= \boldsymbol{C}_n \boldsymbol{x},
\end{align}
where 
\begin{equation}
\boldsymbol{A}_n = \left[\begin{smallmatrix} \boldsymbol{0}_{6\times6} & \boldsymbol{I}_{6} \\ \boldsymbol{0}_{6\times6} & \boldsymbol{0}_{6\times6} \end{smallmatrix}\right], ~
\boldsymbol{B}_n = \left[\begin{smallmatrix} \boldsymbol{0}_{6\times6}\\ \boldsymbol{M}_n^{-1} \end{smallmatrix}\right], ~
\boldsymbol{C}_n = \left[\begin{smallmatrix} \boldsymbol{I}_{6} \\ \boldsymbol{0}_{6\times6} \end{smallmatrix}\right]^\top.
\end{equation}
% \begin{equation}
% \boldsymbol{A}_n
% =
% \begin{bmatrix}
% \boldsymbol{0}_{6\times6} & \boldsymbol{I}_{6} \\
% \boldsymbol{0}_{6\times6} & \boldsymbol{0}_{6\times6}
% \end{bmatrix},
% \boldsymbol{B}_n
% =
% \begin{bmatrix}
% \boldsymbol{0}_{6\times6}\\
% \boldsymbol{M}_n^{-1}
% \end{bmatrix},
% \boldsymbol{C}_n
% =
% \begin{bmatrix}
% \boldsymbol{I}_{6} \\ \boldsymbol{0}_{6\times6}
% \end{bmatrix}^\top.    
% \end{equation}

\begin{assumption}
\label{assumption_boundedness_lumped_uncertainties}
For the system in \eqref{eqn:dynamics}, the environmental disturbance
$\boldsymbol{\tau}_{\mathrm{wave}}(t)$ is continuously differentiable, and both
$\boldsymbol{\tau}_{\mathrm{wave}}(t)$ and its time derivative are uniformly
bounded. That is, there exist constants $\bar{\tau}_{\mathrm{wave}}>0$ and
$\bar{d}_{\mathrm{wave}}>0$ such that, for all $t\geq 0$,
\begin{equation}
\left\|\boldsymbol{\tau}_{\mathrm{wave}}(t)\right\|
\leq \bar{\tau}_{\mathrm{wave}},
\qquad
\left\|\dot{\boldsymbol{\tau}}_{\mathrm{wave}}(t)\right\|
\leq \bar{d}_{\mathrm{wave}}.
\end{equation}
\end{assumption}
\begin{remark}
\label{remark:boundedness}
The model functions $\boldsymbol{M}$,
$\boldsymbol{C}_{\mathrm{Cor}}$, $\boldsymbol{D}$, and $\boldsymbol{g}$
are continuously differentiable on their entire domains, and
$\boldsymbol{M}(\boldsymbol{\eta})$ is nonsingular for all
$\boldsymbol{\eta}\in\mathbb{R}^{6}$. Consequently, the lumped uncertainty
\[
\boldsymbol{\sigma}
=
\boldsymbol{\sigma}
\bigl(
\boldsymbol{x},
\boldsymbol{\tau},
\boldsymbol{\tau}_{\mathrm{wave}}
\bigr)
\]
is continuously differentiable with respect to its arguments. Therefore, for any positive quintuple
$
(\Delta,r,d_r,\bar{\tau}_{\mathrm{wave}},\bar d_{\mathrm{wave}})
$, 
there exist positive constants
\[
L_{\sigma}
=
L_{\sigma}(\Delta,r,\bar{\tau}_{\mathrm{wave}})
\]
and
\[
L_{\dot{\sigma}}
=
L_{\dot{\sigma}}
(\Delta,r,d_r,\bar{\tau}_{\mathrm{wave}},\bar d_{\mathrm{wave}})
\]
such that
\begin{equation}
\|\boldsymbol{\sigma}\|\leq L_{\sigma},
\qquad
\|\dot{\boldsymbol{\sigma}}\|\leq L_{\dot{\sigma}},
\end{equation}
whenever $\|\boldsymbol{x}\|\leq\Delta$, $\|\boldsymbol{\tau}\|\leq r$, $\|\dot{\boldsymbol{\tau}}\|\leq d_r$, $\|\boldsymbol{\tau}_{\mathrm{wave}}\|
\leq\bar{\tau}_{\mathrm{wave}}$,  and $\|\dot{\boldsymbol{\tau}}_{\mathrm{wave}}\|
\leq\bar d_{\mathrm{wave}}$. \hfill $\circ$
\end{remark}

The control input is subject to practical actuator constraints on both
its magnitude and rate of change. 
Specifically,
\begin{equation}
\label{eqn:sat_par}
|\tau^i(t)| \le \tau_{\max}^i, 
~ 
|\dot{\tau}^i(t)|
\le
{\tau}^i_{\max,d},
~
i\in\{x,y,z,\phi,\theta,\psi\},    
\end{equation}
where $\tau_{\max}^i$ and ${\tau}^i_{\max,d}$ denote the maximum actuator capability and the maximum achievable input rate in the $i$th degree of freedom, respectively.

\begin{remark}
Lyapunov analysis\cite{nonlinear} subsequently shows that the proposed constrained control
law renders the closed-loop state semi-globally practically bounded. \hfill $\circ$
\end{remark}

The objective is to design an input-constrained ADRC scheme, consisting of an Extended State Observer (ESO), a state-feedback control law, and a low-pass filter, such that:
\begin{enumerate}
    \item the closed-loop system is stabilized at the origin in the presence of $\boldsymbol{\sigma}$;
    \item the control input satisfies the actuator constraints;
    \item the control input derivative remains bounded;
    \item the closed-loop system is semi-globally practically asymptotically stable.
\end{enumerate}

To ensure closed-loop well-posedness, a saturation map and a low-pass
filter enforce the input-magnitude and input-rate constraints,
respectively. The feedback gains are chosen such that saturation remains
inactive for all initial conditions in a prescribed compact set. The
resulting fast-observer, intermediate-filter, and slow-plant dynamics
enable a singular perturbation stability analysis.

%%%%%%%%%%%%%%%%%%%%%%%%%%%%%%%%%%%%%%%%%%%%%%%%%%%%%%%%%%%%%%%%%%%%%%%%%%%%%%%%%%%%%%%%%%%%

\section{Method}
\label{sec:method}

This section presents the proposed input-constrained active disturbance
rejection control (IC-ADRC) framework, shown in
Fig.~\ref{fig:framework}. An extended state observer (ESO) estimates the states
and lumped uncertainty, which is compensated by the ADRC law. An active
saturation map and a low-pass filter enforce the input-magnitude and
input-rate constraints, respectively. The closed-loop system is then
reformulated to reveal the time-scale separation required for singular
perturbation analysis.
\begin{figure*}
    \centering
    \includegraphics[width=0.7\linewidth]{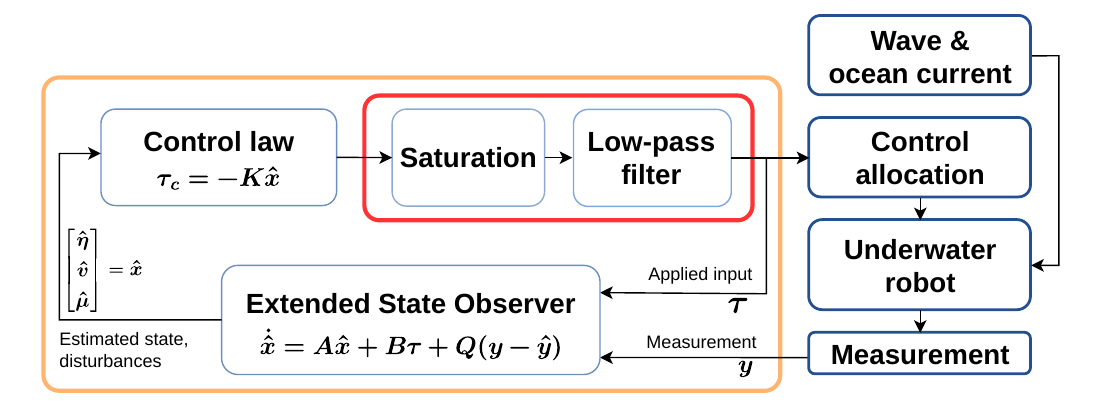}
    \caption{Input-constrained ADRC framework}
    \label{fig:framework}
\end{figure*}

\subsection{ESO Design}
\label{subsec:eso_design}

To estimate $\boldsymbol{\sigma}$, we define the state $\boldsymbol{\mu}$ and derivative $\boldsymbol{\xi}$ as
\begin{equation}
    \boldsymbol{\mu} := \boldsymbol{M}_n^{-1}\boldsymbol{\sigma}, 
    \qquad 
    \boldsymbol{\xi} := \dot{\boldsymbol{\mu}} = \boldsymbol{M}_n^{-1}\dot{\boldsymbol{\sigma}}.
    \label{eqn:mu_def}
\end{equation}
Since $\boldsymbol{M}_n$ is constant and invertible, boundedness of
$\boldsymbol{\sigma}$ and $\dot{\boldsymbol{\sigma}}$ implies boundedness
of $\boldsymbol{\mu}$ and $\boldsymbol{\xi}$. Define
$\boldsymbol{z}=[\boldsymbol{\eta}^{\top},\boldsymbol{v}^{\top},
\boldsymbol{\mu}^{\top}]^{\top}\in\mathbb{R}^{18}$, yielding
\begin{equation}
    \dot{\boldsymbol{z}}
    =
    \boldsymbol{A}\boldsymbol{z}
    +
    \boldsymbol{B}\boldsymbol{\tau}
    +
    \boldsymbol{E} \boldsymbol{\xi},
    \qquad
    \boldsymbol{y}
    =
    \boldsymbol{C}\boldsymbol{z},
    \label{eqn:extended_system}
\end{equation}
where
\begin{equation}
\begin{aligned}
\boldsymbol{A}
&=
\begin{bmatrix}
\boldsymbol{0}_{6\times6} & \boldsymbol{I}_{6} & \boldsymbol{0}_{6\times6}\\
\boldsymbol{0}_{6\times6} & \boldsymbol{0}_{6\times6} & \boldsymbol{I}_{6}\\
\boldsymbol{0}_{6\times6} & \boldsymbol{0}_{6\times6} & \boldsymbol{0}_{6\times6}
\end{bmatrix},
&
\boldsymbol{B}
&=
\begin{bmatrix}
\boldsymbol{0}_{6\times6}\\
\boldsymbol{M}_n^{-1}\\
\boldsymbol{0}_{6\times6}
\end{bmatrix},
\\
\boldsymbol{C}
&=
\begin{bmatrix}
\boldsymbol{I}_{6} &
\boldsymbol{0}_{6\times6} &
\boldsymbol{0}_{6\times6}
\end{bmatrix},
&
\boldsymbol{E}
&=
\begin{bmatrix}
\boldsymbol{0}_{6\times6}\\
\boldsymbol{0}_{6\times6}\\
\boldsymbol{I}_{6}
\end{bmatrix}.
\end{aligned}
\label{eqn:ABC_matrices}
\end{equation}

Denote the estimated state as $\hat{\boldsymbol{z}} = [\hat{\boldsymbol{\eta}}^\top,~ \hat{\boldsymbol{v}}^\top,~ \hat{\boldsymbol{\mu}}^\top]\in \mathbb{R}^{18}$, the Extended State Observer (ESO) is given by
\begin{equation}
    \dot{\hat{\boldsymbol{z}}}
    =
    \boldsymbol{A}\hat{\boldsymbol{z}}
    +
    \boldsymbol{B}\boldsymbol{\tau}
    +
    \boldsymbol{Q}
    \left(
        \boldsymbol{y}
        -
        \boldsymbol{C}\hat{\boldsymbol{z}}
    \right),
    \label{eqn:observer}
\end{equation}
where the observer gains $\boldsymbol{Q}\in\mathbb{R}^{18\times6}$ are parameterized by a prescribed observer bandwidth $\lambda>0$:
\begin{equation}
    \boldsymbol{Q}
    =
    \begin{bmatrix}
        3\lambda\boldsymbol{I}_6 &
        3\lambda^2\boldsymbol{I}_6 &
        \lambda^3\boldsymbol{I}_6
    \end{bmatrix}^\top.
    \label{eqn:observer_gain}
\end{equation}

\subsection{ADRC Design}
\label{subsec:adrc_design}

Using the ESO estimate $\hat{\boldsymbol{z}}$, we design the commanded input
\begin{equation}
    \boldsymbol{\tau}_c
    =
    -
    \boldsymbol{K}
    \hat{\boldsymbol{z}},
    \label{eqn:tau_c}
\end{equation}
where the control gains
\begin{equation}
    \boldsymbol{K}
    =
    \begin{bmatrix}
        \boldsymbol{K}_P &
        \boldsymbol{K}_D &
        \boldsymbol{M}_n
    \end{bmatrix}
    \in \mathbb{R}^{6\times18}.
    \label{eqn:control_gain}
\end{equation}
The control gains are parameterized by a control bandwidth $\kappa>0$:
\begin{equation}
\label{eqn:kappa}
    \boldsymbol{K}_P
    =
    \kappa^2\boldsymbol{M}_n,
    \qquad
    \boldsymbol{K}_D
    =
    2\kappa\boldsymbol{M}_n.
\end{equation}

The commanded input is saturated using \eqref{eqn:vector_saturation_def} with the limits defined in \eqref{eqn:sat_par},
\begin{equation}
    \boldsymbol{\tau}_s
    =
    \mathrm{sat}
    \left(
    \boldsymbol{\tau}_c,
    \boldsymbol{\tau}_{\max}
    \right).
    \label{eqn:sat}
\end{equation}
In the compact region considered in
Section~\ref{subsec:closed_loop_timescale}, the commanded input remains
strictly within the actuator limits; hence, the saturation is inactive,
\begin{equation}
    \boldsymbol{\tau}_s(t)
    =
    \boldsymbol{\tau}_c(t),
    \qquad
    t\ge t_0.
    \label{eqn:unsat_condition}
\end{equation}
The role of the saturation map is therefore to guarantee actuator feasibility during practical implementation,
while the theoretical analysis characterizes the nominal closed-loop behaviour within the prescribed operating region.

\subsection{Low-pass Filter}
\label{subsec:lpf}

The saturated input $\boldsymbol{\tau}_s$ is smoothed by a first-order
low-pass filter with cutoff frequency $\omega>0$:
\begin{equation}
    \dot{\boldsymbol{\tau}}_f
    =
    \omega
    \left(
        \boldsymbol{\tau}_s
        -
        \boldsymbol{\tau}_f
    \right),
    \label{eqn:filter}
\end{equation}
where $\boldsymbol{\tau}_f$ denotes the filtered control input.
The filtered input $\boldsymbol{\tau}=\boldsymbol{\tau}_f$ is applied to
the vehicle. We define the filter error as
\begin{equation}
    \tilde{\boldsymbol{\tau}}_f
    =
    \boldsymbol{\tau}_f
    -
    \boldsymbol{\tau}_s .
    \label{eqn:filter_error}
\end{equation}

\subsection{Closed-loop Error Dynamics}
\label{subsec:closed_loop_timescale}

This subsection formulates the coupled plant, filter-error, and ESO-error
dynamics and establishes their three-time-scale separation for singular
perturbation stability analysis.
\subsubsection{ESO estimation error}
Define the ESO estimation error
$\tilde{\boldsymbol{z}}=\boldsymbol{z}-\hat{\boldsymbol{z}}$, whose dynamics are obtained by subtracting \eqref{eqn:extended_system} from \eqref{eqn:observer}.
For stability analysis, we define the scaled error
\begin{equation}
    \boldsymbol{\zeta}
    =
    \boldsymbol{D}_{\lambda}
    \tilde{\boldsymbol z},
    \qquad
    \boldsymbol{D}_{\lambda}
    =
    \operatorname{diag}
    \left(
        \lambda^2\boldsymbol I_6,
        \lambda\boldsymbol I_6,
        \boldsymbol I_6
    \right).
    \label{eqn:scaled_eso_error}
\end{equation}
Hence, the dynamics of $\boldsymbol{\zeta}$ can be expressed as
\begin{equation}
    \dot{\boldsymbol{\zeta}}
    =
    \lambda\boldsymbol H\boldsymbol{\zeta}
    +
    \boldsymbol E\boldsymbol{\xi},
    \label{eqn:scaled_eso_dyn}
\end{equation}
where the scaled estimation-error matrix $\boldsymbol H$ is given by
\begin{equation}
    \boldsymbol H
    =
    \begin{bmatrix}
        -3\boldsymbol I_6 & \boldsymbol I_6 & \boldsymbol 0_{6\times6}\\
        -3\boldsymbol I_6 & \boldsymbol 0_{6\times6} & \boldsymbol I_6\\
        -\boldsymbol I_6  & \boldsymbol 0_{6\times6} & \boldsymbol 0_{6\times6}
    \end{bmatrix}.
    \label{eqn:H_E}
\end{equation}

\subsubsection{Plant regulation dynamics}
Since the augmented state $\boldsymbol{z}$ includes the time-varying
disturbance $\boldsymbol{\mu}$, stability is analyzed using the regulation
state $\boldsymbol{x}$ defined in
\eqref{eqn:regulation_state}--\eqref{eqn:linear_system}. 

Expressing \eqref{eqn:tau_c} in terms of $\tilde{\boldsymbol{z}}$ yields
\begin{equation}
\begin{aligned}
    \boldsymbol{\tau}_c
    &=
    -\boldsymbol{K}_P\boldsymbol{\eta}
    -
    \boldsymbol{K}_D\boldsymbol{v}
    -
    \boldsymbol{M}_n\boldsymbol{\mu}
    +
    \boldsymbol{K}\tilde{\boldsymbol{z}}.
    \label{eqn:tau_c_expanded}
\end{aligned}
\end{equation}
Substituting \eqref{eqn:mu_def}, \eqref{eqn:unsat_condition}, \eqref{eqn:filter_error}, and \eqref{eqn:tau_c_expanded} into \eqref{eqn:linear_system} leads to
\begin{equation}
    \boldsymbol{M}_n\dot{\boldsymbol{v}}
    =
    -
    \boldsymbol{K}_P\boldsymbol{\eta}
    -
    \boldsymbol{K}_D\boldsymbol{v}
    +
    \boldsymbol{K}\tilde{\boldsymbol{z}}
    +
    \tilde{\boldsymbol{\tau}}_f .
    \label{eqn:plant_error_preliminary}
\end{equation}

We define the modified filter error as
\begin{equation}
    \bar{\boldsymbol{\tau}}
    :=
    \tilde{\boldsymbol{\tau}}_f
    +
    \boldsymbol{K}\tilde{\boldsymbol{z}}.
    \label{eqn:tau_bar_def}
\end{equation}
Consequently, the plant regulation dynamics become
\begin{equation}
    \dot{\boldsymbol{x}}
    =
    \boldsymbol{A}_c\boldsymbol{x}
    +
    \boldsymbol{B}_c\bar{\boldsymbol{\tau}},
    \label{eqn:control_dyn}
\end{equation}
where
\begin{equation}
    \boldsymbol{A}_c
    =
    \begin{bmatrix}
        \boldsymbol{0}_{6\times6} &
        \boldsymbol{I}_6
        \\
        -\boldsymbol{M}_n^{-1}\boldsymbol{K}_P &
        -\boldsymbol{M}_n^{-1}\boldsymbol{K}_D
    \end{bmatrix},
    \quad
    \boldsymbol{B}_c
    =
    \begin{bmatrix}
        \boldsymbol{0}_{6\times6}
        \\
        \boldsymbol{M}_n^{-1}
    \end{bmatrix}.
    \label{eqn:Ac_Bc}
\end{equation}

\subsubsection{Modified filter error}
Differentiating the modified filter error \eqref{eqn:tau_bar_def} gives
\begin{equation}
    \dot{\bar{\boldsymbol{\tau}}}
    =
    \dot{\tilde{\boldsymbol{\tau}}}_f
    +
    \boldsymbol K\dot{\tilde{\boldsymbol z}}.
    \label{eqn:tau_bar_dot_1}
\end{equation}
Within the unsaturated region, differentiating
\eqref{eqn:tau_c} and \eqref{eqn:filter_error} and using
\eqref{eqn:filter} and \eqref{eqn:unsat_condition} yield
\begin{equation}
    \dot{\tilde{\boldsymbol{\tau}}}_f
    =
    -\omega\tilde{\boldsymbol{\tau}}_f
    +
    \boldsymbol K\dot{\boldsymbol z}
    -
    \boldsymbol K\dot{\tilde{\boldsymbol z}}.
    \label{eqn:tau_f_tilde_dyn}
\end{equation}
By substituting \eqref{eqn:tau_f_tilde_dyn} and
\eqref{eqn:tau_bar_def} into \eqref{eqn:tau_bar_dot_1}, we have
\begin{equation}
    \dot{\bar{\boldsymbol{\tau}}}
    =
    -\omega\bar{\boldsymbol{\tau}}
    +
    \omega\boldsymbol K\tilde{\boldsymbol z}
    +
    \boldsymbol K\dot{\boldsymbol z}.
    \label{eqn:tau_bar_dyn}
\end{equation}

For compactness, we define
\begin{equation}
    \boldsymbol K_x
    =
    \begin{bmatrix}
        \boldsymbol K_P & \boldsymbol K_D
    \end{bmatrix}.
    \label{eqn:Kx_def}
\end{equation}
Using the plant regulation dynamics
\eqref{eqn:control_dyn}, it follows that
\begin{equation}
    \boldsymbol K\dot{\boldsymbol z}
    =
    \boldsymbol K_x\boldsymbol A_c\boldsymbol x
    +
    \boldsymbol K_x\boldsymbol B_c\bar{\boldsymbol\tau}
    +
    \boldsymbol M_n\boldsymbol\xi.
    \label{eqn:K_z_dot}
\end{equation}
Thus, \eqref{eqn:tau_bar_dyn} becomes
\begin{equation}
\begin{aligned}
    \dot{\bar{\boldsymbol{\tau}}}
    ={}&
    -\omega\bar{\boldsymbol{\tau}}
    +
    \omega\boldsymbol K\tilde{\boldsymbol z}
    +\boldsymbol K_x\boldsymbol A_c\boldsymbol x
    +
    \boldsymbol K_x\boldsymbol B_c\bar{\boldsymbol\tau}
    +
    \boldsymbol M_n\boldsymbol\xi .
\end{aligned}
\label{eqn:tau_bar_dyn_expanded}
\end{equation}

\subsubsection{Closed-loop dynamics}

The closed-loop error dynamics comprise
\eqref{eqn:control_dyn}, \eqref{eqn:tau_bar_dyn_expanded}, and
\eqref{eqn:scaled_eso_dyn}, with time scales governed by $\kappa$,
$\omega$, and $\lambda$, respectively.
We define the filter and observer time-scale parameters
\begin{equation}
    \epsilon_2
    =
    \frac{1}{\omega},
    \qquad
    \epsilon_1
    =
    \frac{\omega}{\lambda},
    \qquad
    \epsilon_1\epsilon_2
    =
    \frac{1}{\lambda}.
    \label{eqn:epsilon_def}
\end{equation}
Hence, the coupled error dynamics can be expressed as
\begin{equation}
\label{eqn:total_error}
\begin{aligned}
    \dot{\boldsymbol x}
    ={}&
    \boldsymbol A_c\boldsymbol x
    +
    \boldsymbol B_c\bar{\boldsymbol\tau},
    \\[1mm]
    \epsilon_2\dot{\bar{\boldsymbol\tau}}
    ={}&
    -\bar{\boldsymbol\tau}
    +
    \boldsymbol K\boldsymbol D_\lambda^{-1}
    \boldsymbol\zeta
    \\
    &+
    \epsilon_2
    \left(
        \boldsymbol K_x\boldsymbol A_c\boldsymbol x
        +
        \boldsymbol K_x\boldsymbol B_c\bar{\boldsymbol\tau}
        +
        \boldsymbol M_n\boldsymbol\xi
    \right),
    \\[1mm]
    \epsilon_1\epsilon_2\dot{\boldsymbol\zeta}
    ={}&
    \boldsymbol H\boldsymbol\zeta
    +
    \epsilon_1\epsilon_2
    \boldsymbol E\boldsymbol\xi .
\end{aligned}
\end{equation}

We define the closed-loop error state:
\begin{equation}
    \boldsymbol\chi
    =
    \begin{bmatrix}
        \boldsymbol x^\top &
        \bar{\boldsymbol\tau}^{\top} &
        \boldsymbol\zeta^\top
    \end{bmatrix}^{\top}.
    \label{eqn:xi_def}
\end{equation}

\begin{theorem}
\label{Theorem_Main}
Consider the closed-loop error system \eqref{eqn:total_error}
constructed from the AUV dynamics
\eqref{eqn:kinematics}--\eqref{eqn:dynamics}, the ESO
\eqref{eqn:observer}, the ADRC law \eqref{eqn:tau_c}, and the
low-pass filter \eqref{eqn:filter}. Suppose that
Assumption~\ref{assumption_boundedness_lumped_uncertainties} holds and
that, on the prescribed compact set, the disturbance-compensation term
has a strict actuator-authority margin; that is, there exists $\delta>0$
such that
\[
\left|
[\boldsymbol{M}_n\hat{\boldsymbol{\mu}}(t)]_i
\right|
\leq \tau_{\max}^{i}-\delta,
\qquad t\geq t_0.
\qquad
i\in\{x,y,z,\phi,\theta,\psi\}.
\]
For any prescribed compact set of initial conditions
\[
\|\boldsymbol{\chi}(t_0)\|\leq\Delta
\]
and any prescribed practical accuracy $\nu>0$, there exists
$\kappa^*>0$ such that, for every $\kappa\in(0,\kappa^*)$, the gains
$\boldsymbol{K}_P$ and $\boldsymbol{K}_D$ defined in
\eqref{eqn:kappa} render $\boldsymbol{A}_c$ Hurwitz and ensure that the
commanded input remains within the actuator limits on the prescribed
compact set. Furthermore, for every such $\kappa$, there exist
$\epsilon_2^*(\kappa)>0$ and
$\epsilon_1^*(\kappa,\epsilon_2)>0$ such that, for all
\[
0<\epsilon_2<\epsilon_2^*(\kappa),
\qquad
0<\epsilon_1<\epsilon_1^*(\kappa,\epsilon_2),
\]
the solutions satisfy
\begin{equation}
\|\boldsymbol{\chi}(t)\|
\leq
\beta\!\left(
\|\boldsymbol{\chi}(t_0)\|,
t-t_0
\right)
+\nu,
\qquad t\geq t_0,
\label{eqn:sgpas_bound}
\end{equation}
for some $\beta\in\mathcal{KL}$.
\end{theorem}

{\bf Proof.} 
The proof proceeds in three steps: verification of the slow regulation
dynamics, analysis of the ESO and filter boundary-layer systems, and
application of a hierarchical singular perturbation result.

\textit{Step 1: Nominal regulation dynamics.}
Using the gain parameterization in \eqref{eqn:kappa}, the regulation
matrix in \eqref{eqn:control_dyn} becomes
\begin{equation}
    \boldsymbol A_c
    =
    \begin{bmatrix}
        \boldsymbol 0_{6\times6} & \boldsymbol I_6\\
        -\kappa^2\boldsymbol I_6 & -2\kappa\boldsymbol I_6
    \end{bmatrix}.
\end{equation}
All eigenvalues of $\boldsymbol A_c$ are assigned to $-\kappa$ in the left-half plane, rendering the matrix strictly Hurwitz for every $\kappa>0$. Consider the nominal reduced system
\begin{equation}
    \dot{\boldsymbol x}
    =
    \boldsymbol A_c\boldsymbol x,
    \label{eqn:slow_reduced_proof}
\end{equation}
the origin of this system is globally exponentially stable.

We next justify the unsaturated operation considered in the analysis.
The commanded input can be written as
\begin{equation}
\begin{aligned}
    \boldsymbol\tau_c
    ={}&
    -\kappa^2\boldsymbol M_n\hat{\boldsymbol\eta}
    -2\kappa\boldsymbol M_n\hat{\boldsymbol v}
    -\boldsymbol M_n\hat{\boldsymbol\mu}.
\end{aligned}
\label{eqn:tau_c_proof}
\end{equation}
On the prescribed compact set, there exist finite constants
$c_\eta$ and $c_v$ such that
\begin{equation}
    \left\|
        \boldsymbol M_n\hat{\boldsymbol\eta}
    \right\|
    \leq c_\eta,
    \qquad
    \left\|
        \boldsymbol M_n\hat{\boldsymbol v}
    \right\|
    \leq c_v.
\end{equation}
Suppose that the disturbance-compensation component possesses a strict
actuator-authority margin; namely, there exists $\delta>0$ such that
\begin{equation}
    \left|
        \left[
            \boldsymbol M_n\hat{\boldsymbol\mu}(t)
        \right]_i
    \right|
    \leq
    \tau_{\max}^i-\delta,
    \qquad
    i\in\{x,y,z,\phi,\theta,\psi\}.
    \label{eqn:authority_margin_proof}
\end{equation}
Since
\begin{equation}
    \kappa^2c_\eta+2\kappa c_v
    \rightarrow 0
    \qquad\text{as}\qquad
    \kappa\rightarrow0,
\end{equation}
there exists $\kappa^*>0$ such that, for every
$\kappa\in(0,\kappa^*)$,
\begin{equation}
    \kappa^2c_\eta+2\kappa c_v
    \leq\delta.
\end{equation}
It follows from \eqref{eqn:tau_c_proof} and
\eqref{eqn:authority_margin_proof} that
\begin{equation}
    |\tau_c^i(t)|
    \leq
    \tau_{\max}^i,
    \qquad
    i\in\{x,y,z,\phi,\theta,\psi\}.
\end{equation}
Consequently,
\begin{equation}
    \boldsymbol\tau_s(t)
    =
    \boldsymbol\tau_c(t)
\end{equation}
throughout the prescribed operating region.

\textit{Step 2: Fast ESO dynamics.}
Consider the fastest time scale
\begin{equation}
    s
    =
    \frac{t-t_0}{\epsilon_1\epsilon_2}.
\end{equation}
On this time scale, the slower variables
$\boldsymbol x$ and $\bar{\boldsymbol\tau}$ are treated as frozen.
Setting $\epsilon_1\epsilon_2=0$ in the scaled ESO equation of
\eqref{eqn:total_error} gives the boundary-layer system
\begin{equation}
    \frac{d\boldsymbol\zeta}{ds}
    =
    \boldsymbol H\boldsymbol\zeta.
    \label{eqn:fast_bl_proof}
\end{equation}
All eigenvalues of $\boldsymbol H$, defined in \eqref{eqn:H_E}, co-locate at $-1$.
Therefore, $\boldsymbol H$ is Hurwitz and the equilibrium $\boldsymbol\zeta=\boldsymbol 0$ of
\eqref{eqn:fast_bl_proof} is globally exponentially stable.

Since $\boldsymbol{\xi}$ is bounded,
$\epsilon_1\epsilon_2\boldsymbol{E}\boldsymbol{\xi}$ is a non-vanishing
perturbation of order $\epsilon_1\epsilon_2$. Hence,
$\boldsymbol{\zeta}$, and thus $\tilde{\boldsymbol{z}}$, is guaranteed to
converge only to an arbitrarily small neighborhood of the origin for
sufficiently small $\epsilon_1\epsilon_2$.

Setting $\boldsymbol\zeta=\boldsymbol 0$ gives the first reduced system
\begin{equation}
\label{eqn:first_reduced_proof}
\begin{aligned}
    \dot{\boldsymbol x}
    &=
    \boldsymbol A_c\boldsymbol x
    +
    \boldsymbol B_c\bar{\boldsymbol\tau},
    \\
    \epsilon_2\dot{\bar{\boldsymbol\tau}}
    &=
    -\bar{\boldsymbol\tau}
    +
    \epsilon_2
    \left(
        \boldsymbol K_x\boldsymbol A_c\boldsymbol x
        +
        \boldsymbol K_x\boldsymbol B_c
        \bar{\boldsymbol\tau}
        +
        \boldsymbol M_n\boldsymbol\xi
    \right).
\end{aligned}
\end{equation}

\textit{Step 3: Intermediate filter dynamics.}
Introduce the intermediate time scale
\begin{equation}
    r
    =
    \frac{t-t_0}{\epsilon_2}.
\end{equation}
Treating $\boldsymbol x$ as frozen and setting
$\epsilon_2=0$ in the second equation of
\eqref{eqn:first_reduced_proof} gives the filter boundary-layer system
\begin{equation}
    \frac{d\bar{\boldsymbol\tau}}{dr}
    =
    -\bar{\boldsymbol\tau}.
    \label{eqn:filter_bl_proof}
\end{equation}
Since $\boldsymbol{\xi}$ is bounded,
$\epsilon_1\epsilon_2\boldsymbol{E}\boldsymbol{\xi}
=\mathcal{O}(\epsilon_1\epsilon_2)$. Hence, $\boldsymbol{\zeta}$, and thus
$\tilde{\boldsymbol{z}}$, converges to an arbitrarily small neighborhood
of the origin as $\epsilon_1\epsilon_2$ decreases. Therefore, the slow
reduced system \eqref{eqn:slow_reduced_proof}, intermediate boundary-layer
system \eqref{eqn:filter_bl_proof}, and fast boundary-layer system
\eqref{eqn:fast_bl_proof} are exponentially stable. Moreover,
according to Remark 1, 
$\boldsymbol{\xi}$ is uniformly bounded on the prescribed compact set.
Consequently, all conditions of the hierarchical singular perturbation
theorem are satisfied.

It follows from the hierarchical singular perturbation theorem that,
for every prescribed compact set of initial error states and every
$\nu>0$, there exist sufficiently small
$\epsilon_2>0$ and, subsequently, sufficiently small
$\epsilon_1>0$ such that the solutions of the complete error system
\eqref{eqn:total_error} satisfy
\begin{equation}
    \|\boldsymbol\chi(t)\|
    \leq
    \beta
    \left(
        \|\boldsymbol\chi(t_0)\|,
        t-t_0
    \right)
    +
    \nu,
    \qquad
    t\geq t_0,
\end{equation}
for some $\beta\in\mathcal{KL}$.  Therefore, the origin of
$(\boldsymbol{x},\bar{\boldsymbol{\tau}},\boldsymbol{\zeta})$
is semi-globally practically asymptotically stable, completing the proof.
\hfill $\Box$
\begin{remark}
Theorem~\ref{Theorem_Main} guides IC-ADRC tuning. First, select $\kappa$
to achieve the desired response without initial actuator saturation. Next,
choose $\lambda$ sufficiently large for accurate uncertainty estimation
but not so large that it amplifies measurement noise. Finally, select
$\omega$ such that
\[
\kappa\ll\omega\ll\lambda,
\]
while balancing phase lag, noise attenuation, and actuator rate limits.
Fig.~\ref{fig:pole} illustrates this three-time-scale separation.
\hfill $\circ$
\end{remark}

\begin{figure}[htb!]
    \centering
    \includegraphics[width=0.6\linewidth]{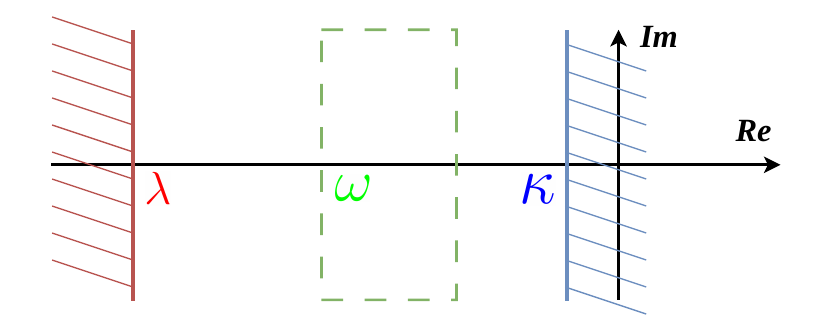}
 \caption{Three-time-scale pole separation under
$\kappa\ll\omega\ll\lambda$.}
    \label{fig:pole}
\end{figure}

%%%%%%%%%%%%%%%%%%%%%%%%%%%%%%%%%%%%%%%%%%%%%%%%%%%%%%%%%%%%%
%%%%%%%%%%%%%%%%%%%%%%%%%%%%%5
\section{Simulation Validation}
\label{sec:sim}

This section evaluates IC-ADRC through four high-fidelity simulations.
Sections~\ref{sec:setup1} and \ref{sec:compare} examine parameter tuning,
post-saturation filtering, violation of time-scale separation, and
comparison with existing controllers, respectively.

The simulations use the fully actuated heavy configuration of BlueROV2,
governed by \eqref{eqn:kinematics} and \eqref{eqn:dynamics} with parameters
from \cite{blue}. The ROS simulator runs at 100~Hz, with the world frame
initialized at the vehicle's initial position.

External disturbances are introduced through oscillatory surface waves
and steady ocean currents to emulate underwater operating conditions.
The wave disturbance is modeled using Airy wave theory \cite{wave}, under
which water particles follow elliptical trajectories determined by the
water depth. Since the vehicle is small relative to the wavelength, it is
approximated as a water particle. The wave propagates in the negative
$x$-direction with wavelength $L=5~\mathrm{m}$, height
$H=2~\mathrm{m}$, period $T=10~\mathrm{s}$, and water depth
$h=10~\mathrm{m}$. The vehicle is initially located $0.5~\mathrm{m}$
below the calm water surface. The steady ocean current is modeled as the constant external force
\begin{equation}
\label{eqn:current}
\boldsymbol{\tau}_{\mathrm{cur}}
=
[-10,\,0,\,0,\,0,\,0,\,0]^\top,
\end{equation}
acting in the direction of wave propagation.

Overall convergence is quantified by the pose error norm (PEN), defined as the Euclidean norm of $\boldsymbol{\eta}$
\begin{equation}
\label{eqn:eta_norm}
e_{\boldsymbol{\eta}}=\|\boldsymbol{\eta}\|.
\end{equation}
Equal weighting is assigned to the translational ($x$, $y$, $z$) and
rotational ($\phi$, $\theta$, $\psi$) components.

The proposed controller is evaluated through two representative control
tasks. The \emph{stabilization} task requires the vehicle to converge from
a nonzero initial displacement to the origin. The
\emph{station-keeping} task evaluates disturbance rejection by maintaining
the vehicle at the origin under environmental disturbances. Together,
these tasks assess transient response and disturbance rejection.

\subsection{Influence of Controller Parameters on Closed-Loop Performance}
\label{sec:setup1}

\begin{figure*}
    \centering
    \includegraphics[width=\textwidth]{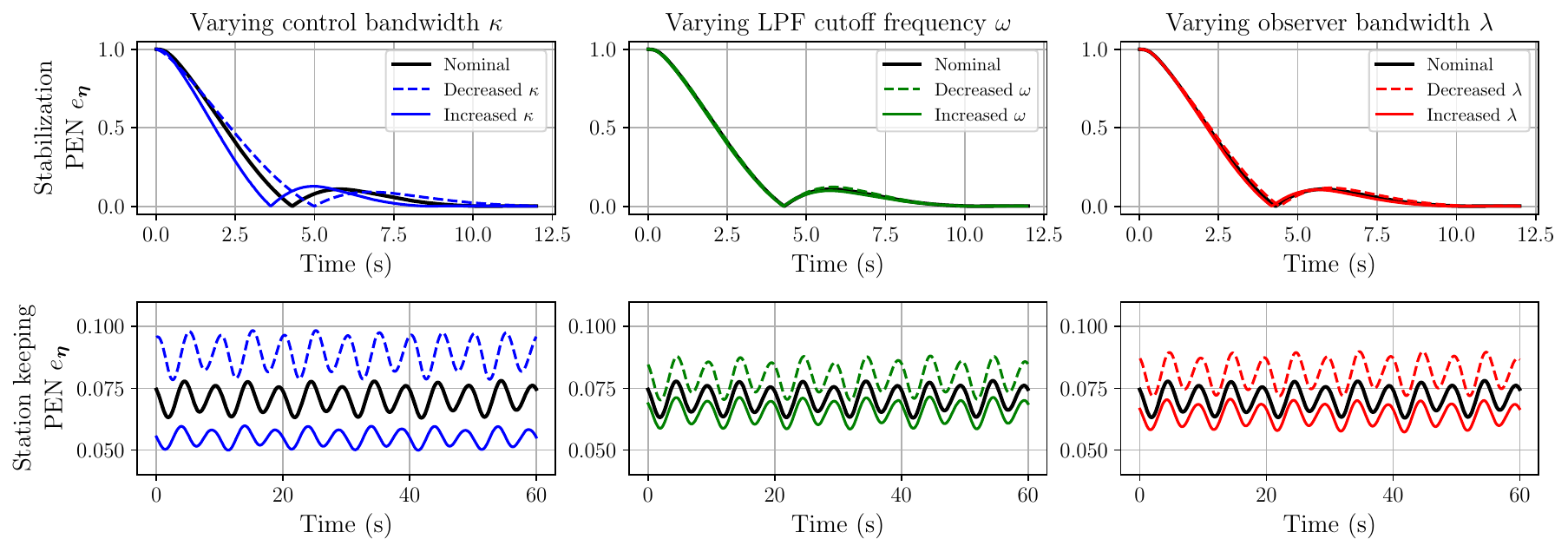}
    \caption{Sensitivity of the pose error to variations in $\kappa$, $\omega$,
and $\lambda$ for stabilization (top) and station keeping (bottom). Solid,
dashed, and black lines denote increased, decreased, and nominal parameter
values, respectively.}
    \label{fig:setup1}
\end{figure*}

This study examines the effect of controller parameter tuning on
closed-loop performance while preserving the time-scale separation required
by Theorem~\ref{Theorem_Main}. The nominal parameters are listed in
Table~\ref{tab:nominal_params}. The saturation limits in \eqref{eqn:sat}
are
\begin{equation}
\label{eqn:sat_value}
\boldsymbol{\tau}_{\max}=[40,40,40,10,10,10]^\top.
\end{equation}

\begin{table}
\centering
\caption{IC-ADRC Parameters}
\label{tab:nominal_params}
\renewcommand{\arraystretch}{1.1}
\begin{tabular}{ll}
\toprule
\textbf{Parameter} & \textbf{Value} \\
\midrule
Control bandwidth & $\kappa=1$ \\
Low-pass filter angular cutoff frequency & $\omega = 4$ \\
Extended State Observer (ESO) bandwidth & $\lambda=10$ \\
\bottomrule
\end{tabular}
\end{table}

To evaluate the sensitivity of the proposed controller, each tuning parameter is independently varied by $\pm20\%$, while all other parameters are held at their nominal values as in Table~\ref{tab:nominal_params}.
The resulting configurations are detailed in Table~\ref{tab:parameter_variation}.

\begin{table}
\centering
\caption{Parameter Variations}
\label{tab:parameter_variation}
\renewcommand{\arraystretch}{1.1}
\begin{tabular}{p{2.2cm}lll}
\toprule
\textbf{Case} & \multicolumn{3}{l}{\textbf{Modified Parameters}} \\
\midrule
Decreased $\kappa$ & $\kappa=0.8$ & $\omega=4$ & $\lambda=10$ \\
Increased $\kappa$ & $\kappa=1.2$ & $\omega=4$ & $\lambda=10$ \\
Decreased $\omega$ & $\kappa=1$ & $\omega=3.2$ & $\lambda=10$ \\
Increased $\omega$ & $\kappa=1$ & $\omega=4.8$ & $\lambda=10$ \\
Decreased $\lambda$ & $\kappa=1$ & $\omega=4$ & $\lambda=8$ \\
Increased $\lambda$ & $\kappa=1$ & $\omega=4$ & $\lambda=12$ \\
\bottomrule
\end{tabular}
\end{table}

The controller is evaluated under two representative operating conditions:
\begin{itemize}
\item Stabilization from an initial displacement of $1~\mathrm{m}$ along the $x$-axis.
\item Station keeping under simulated wave disturbances and ocean currents.
\end{itemize}

Fig.~\ref{fig:setup1} presents the sensitivity analysis for the
stabilization (top row) and station-keeping (bottom row) tasks. Overall,
closed-loop performance improves as the controller bandwidths increase.
Among the three parameters, $\kappa$ has the greatest influence on
convergence and disturbance rejection, whereas variations in $\omega$ and
$\lambda$ produce smaller deviations from the nominal response. These results agree with Theorem~\ref{Theorem_Main}. Under sufficient
time-scale separation, $\kappa$ primarily determines the closed-loop
dynamics, while $\omega$ and $\lambda$ govern control smoothness and
disturbance estimation accuracy, respectively. 
However, practical constraints limit these parameters. Increasing
$\kappa$ produces aggressive control actions, actuator saturation, and
overshoot. Increasing $\omega$ amplifies high-frequency control chattering,
as demonstrated in Section~\ref{sec:setup3}, while an excessively large
$\lambda$ amplifies measurement noise and introduces closed-loop
oscillations.

\subsection{Effect of Low-Pass Filter Selection}
\label{sec:setup2}

\begin{figure}[tb!]
    \centering
    \includegraphics[width=0.8\linewidth]{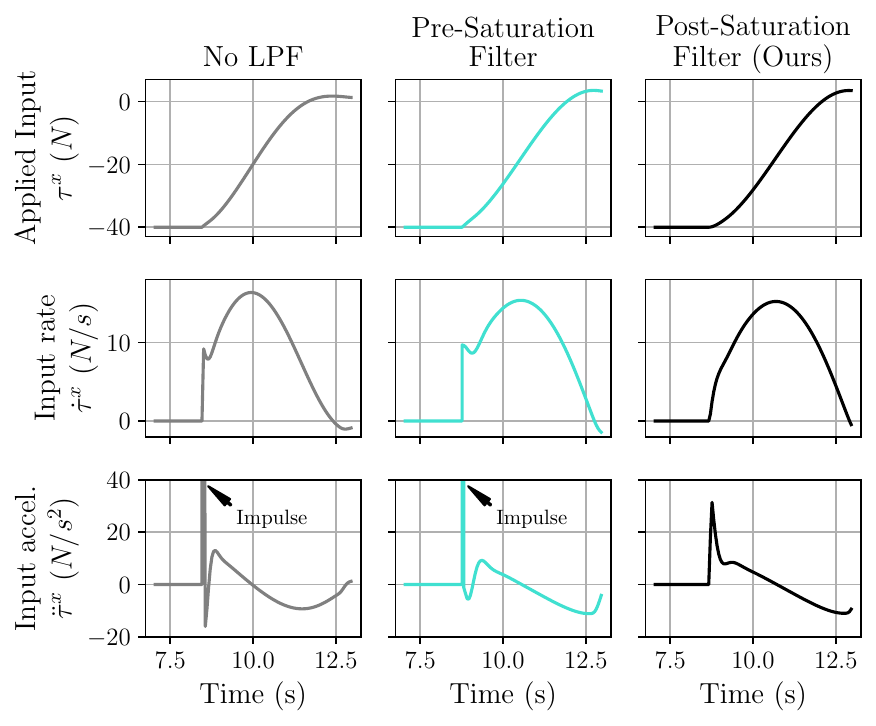}
    \caption{Effect of filter placement at the saturation boundary
($\tau_{\max}^{x}=40$~N): no filtering (left), pre-saturation filtering
(center), and post-saturation filtering (right). Only post-saturation
filtering ensures a continuously differentiable input with bounded rate and
acceleration.}
    \label{fig:setup2}
\end{figure}

Physical actuators are inherently subject to both magnitude and rate constraints. 
Although standard ADRC effectively compensates for disturbances, it does not explicitly enforce these constraints. 
In particular, the commanded input $\boldsymbol{\tau}_c$ \eqref{eqn:tau_c} may become arbitrarily large due to the high-gain observer, motivating the active input saturation introduced in \eqref{eqn:sat}.

While input saturation limits the control magnitude, it introduces a new challenge. 
The hard saturation function is continuous but non-differentiable at the transition boundary $(|{\tau}_c^i|={\tau}_{\max}^i)$. 
Consequently, the control input rate becomes discontinuous, potentially introducing numerical stiffness in simulation and excessive mechanical stress in physical actuators. 
The proposed low-pass filter (LPF) addresses this issue by smoothing the saturated control input, thereby ensuring a continuously differentiable and physically realizable input signal.

To investigate the effect of filter placement, the vehicle is initialized $5~\mathrm{m}$ from the origin in the $x$-axis so that the commanded surge force immediately exceeds the actuator limit, i.e.,
\[
\tau_c^{x}(t_0)>\tau_{\max}^{x}.
\]

Three controller configurations are compared while all remaining parameters are fixed at their nominal values in Table~\ref{tab:nominal_params}:
\begin{itemize}
\item \textbf{Baseline (No LPF):} Active input saturation only.
\item \textbf{Pre-saturation LPF:} The low-pass filter is applied before the saturation block.
\item \textbf{Post-saturation LPF (proposed):} The low-pass filter is applied after the saturation block.
\end{itemize}

Fig.~\ref{fig:setup2} compares the control signals during the transition
from saturated to unsaturated operation. Without filtering (left column),
the control input has a nondifferentiable transition, resulting in a
discontinuous rate $\dot{\tau}^{x}$ and an impulse in the acceleration
$\ddot{\tau}^{x}$. This behavior is physically unrealizable and may excite
actuator dynamics or cause numerical instability. Pre-saturation filtering
(center column) cannot remove the discontinuity because the downstream
saturation reintroduces the corner. In contrast, post-saturation filtering
(right column) smooths the constrained input, ensuring a continuous rate
and bounded acceleration. Thus, only post-saturation filtering preserves
the actuator constraints while providing the input smoothness required for
practical implementation and the stability analysis in
Theorem~\ref{Theorem_Main}.

\subsection{Influence of Time-Scale Separation on Closed-Loop Performance}
\label{sec:setup3}

\begin{figure}[htb!]
    \centering
    \includegraphics[width=0.8\linewidth]{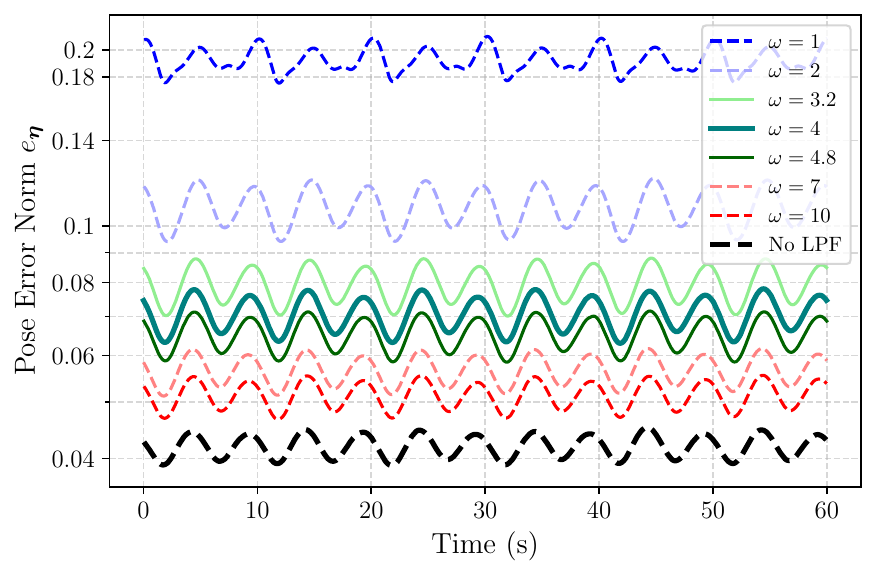}
    \caption{Effect of the LPF cutoff frequency $\omega$ on stabilization.
Solid and dashed lines denote recommended and extreme values, respectively;
the dashed black line represents the unfiltered case.}
    \label{fig:setup3_wave}
\end{figure}  

\begin{figure}[htb!]
    \centering
    \includegraphics[width=0.8\linewidth]{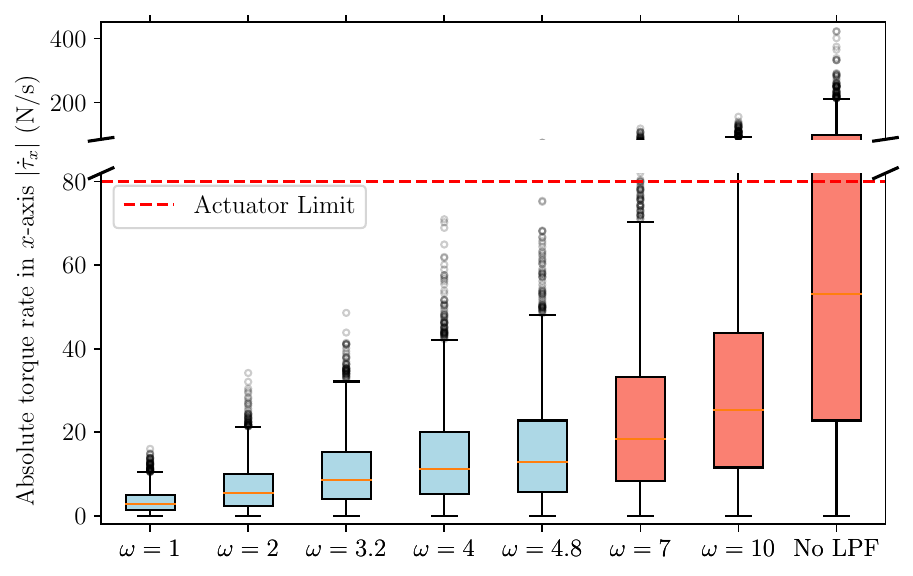}
    % \caption{Control errors in surge and heave under wave disturbances for different angular cutoff frequencies.}
    \caption{Distribution of the absolute surge input rate
$|\dot{\tau}^{x}|$ under measurement noise. The dashed red line denotes
the actuator limit of $80$~N/s; the broken axis shows extreme outliers.}
    \label{fig:setup3_noise}
\end{figure}

This subsection examines the effect of time-scale separation on closed-loop
performance and validates Theorem~\ref{Theorem_Main}. The cutoff frequency
$\omega$ is varied while $\kappa$ and $\lambda$ remain at their nominal
values in Table~\ref{tab:nominal_params}. The tested values are grouped
into four categories:
\begin{itemize}
    \item \textbf{Lower-bound violation}: $\omega=\{1,2\}$, where the filter bandwidth is insufficiently large relative to the controller bandwidth, introducing excessive phase lag.
    
    \item \textbf{Upper-bound violation}: $\omega=\{7,10\}$, where the filter bandwidth is excessively large relative to the observer dynamics, increasing sensitivity to measurement noise.
    
    \item \textbf{Valid tuning}: $\omega=\{3.2,4,4.8\}$, satisfying the required time-scale separation.
    
    \item \textbf{No LPF}: the low-pass filter is removed, i.e.,
    $\boldsymbol{\tau}=\boldsymbol{\tau}_s$.
\end{itemize}
Two scenarios are considered:
\begin{itemize}
    \item \textbf{Wave disturbance}: The vehicle is subjected to the simulated wave disturbances and ocean currents, while perfect state measurements are assumed.
    
    \item \textbf{Measurement noise}: External wave disturbances are removed, while zero-mean Gaussian noise is added to the measurement equation
    \begin{equation}
            \boldsymbol{y}=\boldsymbol{\eta}+\boldsymbol{\varepsilon}, \qquad \boldsymbol{\varepsilon} = [\varepsilon^{x},0,0,0,0,0]^\top ,
    \end{equation}
    where $\varepsilon^{x}\sim\mathcal{N}(0,\sigma^2)$ with $\sigma=0.01~\mathrm{m}$. 
    The noise is truncated at $\pm4\sigma$ to avoid unrealistic outliers.
\end{itemize}

We first examine disturbance rejection under wave excitation.
Fig.~\ref{fig:setup3_wave} shows the pose error
$e_{\boldsymbol{\eta}}$ for varying $\omega$. Within the
admissible range, increasing $\omega$ reduces the filter-induced phase lag,
allowing faster disturbance rejection and reducing the pose error. In
contrast, values below this range introduce excessive delay and degrade
performance by nearly an order of magnitude. The second scenario shows why
$\omega$ cannot be increased arbitrarily. Based on the actuator dynamics
identified in \cite{blue,t200}, a conservative slew time of
$0.5~\mathrm{s}$ is adopted, corresponding to a maximum surge input rate
of $80~\mathrm{N/s}$. As shown in Fig.~\ref{fig:setup3_noise}, removing
the filter or selecting an excessively large $\omega$ amplifies
measurement noise, producing rapid control variations that frequently
exceed this limit.

These results support Theorem~\ref{Theorem_Main}. A small $\omega$
increases phase lag and degrades disturbance rejection, whereas a large
$\omega$ amplifies noise and may violate actuator rate limits. Therefore,
$\omega$ should be chosen between $\kappa$ and $\lambda$ to balance
disturbance rejection, noise attenuation, and actuator feasibility.

\subsection{Comparative Evaluation}
\label{sec:compare}

This comparison evaluates representative control strategies with different
design philosophies under identical operating conditions of the \textit{station-keeping} task. The assessment
considers station-keeping accuracy, disturbance rejection, control effort,
and compliance with input and input-rate constraints. The controllers being compared are:

\begin{itemize}
\item \textbf{PID}: 
A proportional--integral--derivative (PID) controller is included for its simplicity and widespread industrial and marine adoption \cite{Han2009ADRC,fossen}. 
For a fair comparison, the P and D gains are matched to the nominal state-feedback component of IC-ADRC, enforcing identical stiffness and damping, isolating the performance contributions of the ESO and the input-constrained architecture. 
The integral gain is empirically tuned to cancel steady-state errors without degrading transient response, yielding:
\(
\boldsymbol{K}_P=\boldsymbol{M}_n,\quad
\boldsymbol{K}_D=2\boldsymbol{M}_n, \quad
\boldsymbol{K}_I=0.2\boldsymbol{M}_n.
\)
\item \textbf{SMC}:
Sliding mode control is included as a representative nonlinear strategy
with inherent robustness to matched disturbances and parameter
uncertainties \cite{Utkin1992}, both common in underwater robotics
\cite{blue,blue2}. Following \cite{blue,blue2}, a 6-DOF element-wise SMC
with a sigmoid approximation is used to reduce high-frequency switching.
The parameters $\boldsymbol{\Lambda}$ and $\boldsymbol{\epsilon}$ are
retained from \cite{blue,blue2}, while
\[
\boldsymbol{K}=\operatorname{diag}(39,39,39,2,2,2).
\]
Compared with \cite{blue}, the translational gains are increased by
approximately $25\%$ to reject peak disturbances while retaining sufficient
margin below the saturation limit \eqref{eqn:sat_par}.

\item \textbf{NMPC}:
Nonlinear model predictive control is included as a representative
model-based optimal control strategy using a receding-horizon framework
\cite{Grune2017}. It repeatedly optimizes the control sequence using the
current state and a nonlinear prediction model, providing feedback while
explicitly enforcing input constraints. These features make NMPC attractive
for underwater robotics \cite{mpc5,mpc6}. Following \cite{mpc5,mpc6}, the prediction model uses the nominal vehicle
dynamics \eqref{eqn:kinematics}--\eqref{eqn:dynamics}, excluding the unknown
wave disturbance $\boldsymbol{\tau}_{\mathrm{wave}}$. The finite-horizon
problem is implemented in CasADi \cite{casadi}, with
\[
\boldsymbol{Q}=10^5\boldsymbol{I}_{12},\qquad
\boldsymbol{R}=\boldsymbol{I}_{6},
\]
a prediction horizon $N=20$, and a sampling period
$\Delta t=0.1$~s. The optimization uses a fourth-order Runge--Kutta
discretization and enforces the actuator limits
\eqref{eqn:sat_value} and slew time $t_{\mathrm{slew}}=0.5$~s from
Section~\ref{sec:setup2}.  

\item \textbf{IC-ADRC}:
The proposed IC-ADRC combines an ESO, state feedback control law, active input
saturation, and a post-saturation low-pass filter. It preserves ADRC
disturbance compensation while enforcing practical input and input-rate
constraints, providing a computationally efficient solution for underwater
robots with actuator limitations. The parameters are listed in
Table~\ref{tab:nominal_params}.
\end{itemize}

\subsubsection*{Test Cases}

The four controllers are evaluated in the \textit{station-keeping} task under six operating
conditions covering nominal, decomposed, and challenging disturbances. In
all cases, the robot is initially stationary at $z_0=0.5$~m in water of
depth $h=10$~m. The wavelength is fixed at $L=5$~m, while the wave height
$H$, period $T$, and current force $\tau_{\mathrm{cur}}^x$ along the
negative $x$-axis vary as listed in Table~\ref{tab:testcases}.
\begin{table}[htb!]
\renewcommand{\arraystretch}{1.25} 
\centering
\caption{Simulation test cases for controller comparison.}
\label{tab:testcases}
\begin{tabular}{|c|l|c|c|c|}
\hline
\textbf{Wave} & \textbf{Description}  \rule[-6pt]{0pt}{16pt}  & $H$ (m) & $T$ (s) & $\tau_{\mathrm{cur}}^x$ (N)\\
\hline
1 & Nominal              & 2.0 & 10 & $-10$\\\hline
2 & Wave only           & 2.0 & 10 & 0\\ \hline
3 & Current only         & 0   & -- & $-10$\\ \hline
4 & High-frequency wave  & 0.1 & 2  & $-10$\\ \hline
5 & High-amplitude wave  & 6.0 & 10 & $-10$\\ \hline
6 & Strong current       & 2.0 & 10 & $-30$\\ 
\hline
\end{tabular}
\end{table}

\subsubsection*{Performance Metrics}

Controller performance is evaluated using four metrics for tracking
accuracy, control effort, and actuator feasibility.

\begin{itemize}
\item \textbf{Average Pose Error Norm (APEN):}
The mean pose error over the simulation horizon is
\begin{equation}
\label{eqn:avg_error_discrete}
\bar{e}_{\boldsymbol{\eta}}
=
\frac{1}{N}\sum_{k=1}^{N}e_{\boldsymbol{\eta}}(t_k),
\end{equation}
where $e_{\boldsymbol{\eta}}$ is defined in \eqref{eqn:eta_norm}. APEN
measures overall station-keeping accuracy.
\item \textbf{Maximum Pose Error Norm (MPEN):}
The maximum pose error is
\begin{equation}
\label{eqn:max_error_discrete}
e_{\boldsymbol{\eta}}^{\max}
=
\max_{k=1,\ldots,N}e_{\boldsymbol{\eta}}(t_k),
\end{equation}
measuring the largest deviation from the desired pose.
\item \textbf{Average Control Effort (ACE):}
The average control effort is
\begin{equation}
\label{eqn:avg_energy}
\bar{E}_{\boldsymbol{\tau}}
=
\frac{1}{N}\sum_{k=1}^{N}
\left\|\boldsymbol{\tau}(t_k)\right\|,
\end{equation}
which measures the mean actuator demand during station keeping. A lower
ACE is preferable for battery-powered underwater vehicles.
\item \textbf{Maximum Input Rate (MIR):}
Actuator feasibility is assessed using the maximum input rate over the
translational degrees of freedom:
\begin{equation}
\label{eqn:max_input_rate}
\dot{\tau}_{\max}
=
\max_{\substack{k=1,\ldots,N-1\\j\in\{x,y,z\}}}
\left|
\frac{\tau_j(t_{k+1})-\tau_j(t_k)}{\Delta t}
\right|,
\end{equation}
where $\Delta t=t_{k+1}-t_k$. This metric evaluates compliance with the
slew-rate constraint in Section~\ref{sec:setup3}.

\end{itemize}

Table~\ref{tab:compare} compares the four controllers under all six
operating conditions. Several observations follow.

First, disturbance compensation strongly affects station-keeping accuracy.
Under wave-only conditions (Wave~2), SMC achieves the lowest APEN and MPEN,
while NMPC also performs well through model-based prediction. However,
because the NMPC model excludes unknown environmental disturbances, its
accuracy deteriorates when persistent currents are introduced. In contrast,
IC-ADRC performs particularly well in the current-dominated cases
(Waves~3 and 6) by estimating and compensating the lumped uncertainty
through the ESO.

Second, SMC achieves the smallest tracking errors in several wave-dominated
cases but requires aggressive control action. Despite the sigmoid boundary
layer, its maximum input rate is more than one order of magnitude higher
than that of IC-ADRC in all cases, limiting its practicality for
actuator-constrained underwater vehicles.

Third, PID reduces steady-state errors under constant disturbances through
integral action. However, its limited bandwidth weakens its rejection of
rapidly varying waves, resulting in larger APEN and MPEN in most
wave-dominated cases.

Overall, IC-ADRC provides a favourable balance among tracking accuracy,
disturbance rejection, control effort, and actuator feasibility. By
combining online disturbance estimation with explicit input and input-rate
constraints, it achieves competitive tracking with the lowest actuator
slew rates across all operating conditions.

\begin{table}[htb!]
\centering
\caption{Controller performance under varying wave conditions. Top-performing results are shown in bold.}
\label{tab:compare}
\renewcommand{\arraystretch}{1.0} 
% \begin{tabular}{\textwidth}{@{\extracolsep{\fill}} l l c c c c @{}}
\begin{tabular}{l l c c c c}
\toprule
\textbf{Wave Settings} & \textbf{Metrics} & \textbf{NMPC} & \textbf{SMC} & \textbf{PID} & \makecell{\textbf{IC-ADRC} \\ \textbf{(Ours)}} \\
\midrule

% --- Group 1: The Baseline ---
\multirow{4}{*}{Wave 1} 
& APEN $\downarrow$ & 0.0631          & \textbf{0.0291} & 0.1236 & 0.0682 \\
& MPEN $\downarrow$ & {0.0908}        & \textbf{0.0447}  & 0.4453 & 0.1449 \\
& ACE $\downarrow$  & \textbf{11.16}         & 16.85           & \textbf{11.23}  & \textbf{11.32} \\
& MIR $\downarrow$   & 56.9           & 962.2           & 201.0  & \textbf{26.3} \\

\midrule
\multicolumn{6}{l}{\textit{Disturbance Decomposition}} \\
\midrule

\multirow{4}{*}{ \makecell[cl]{Wave 2 \\ (Wave only)} }
& APEN $\downarrow$ & {0.0185}      & \textbf{0.0073}   & 0.1155        & 0.0673 \\
& MPEN $\downarrow$ & {0.0193}      & \textbf{0.0109}   & 0.1396        & 0.0704 \\
& ACE $\downarrow$  & \textbf{5.56}   & 9.85            & \textbf{5.20} & \textbf{5.60} \\
& MIR $\downarrow$  & 44.9            & 950.9           & 159.6         & \textbf{15.2} \\
\addlinespace
\multirow{4}{*}{\makecell[cl]{Wave 3 \\ (Current only)}} 
& APEN $\downarrow$ & 0.0596         &  0.0275          & 0.0144   & \textbf{0.0016} \\
& MPEN $\downarrow$ & {0.0599}      & \textbf{0.0285}   & 0.3260   & {0.1126} \\
& ACE $\downarrow$  & \textbf{9.99}        & 15.43           & \textbf{10.02 }  & \textbf{9.98} \\
& MIR $\downarrow$  & 46.2            & 932.3           & 158.1   & \textbf{24.5} \\

\midrule
\multicolumn{6}{l}{\textit{Extreme Robustness Tests}} \\
\midrule

\multirow{4}{*}{ \makecell[cl]{Wave 4 \\ (High-freq)}} 
& APEN $\downarrow$ & {0.0643}          & \textbf{0.0302} & 0.0335          & \textbf{0.0305} \\
& MPEN $\downarrow$ & {0.0787}   & \textbf{0.0479}         & 0.3404         & {0.1340} \\
& ACE $\downarrow$   & 11.92            & 16.45          & \textbf{11.26}  &  \textbf{11.41} \\
& MIR $\downarrow$  & 63.1              & 1115.7         & 175.5           & \textbf{40.3} \\
\addlinespace
\multirow{4}{*}{ \makecell[cl]{Wave 5 \\ (High-amp)} } 
& APEN $\downarrow$ & 0.1041            & \textbf{0.0478} & 0.3001          & 0.2087 \\
& MPEN $\downarrow$ & 0.1829            & \textbf{0.1097} & 0.7631          & 0.2390 \\
& ACE $\downarrow$  & 18.84             & 25.19           & \textbf{17.13} & 21.97 \\
& MIR $\downarrow$  & 79.1              & 933.4           & 319.2          & \textbf{69.4} \\
\addlinespace
\multirow{4}{*}{ \makecell[cl]{Wave 6 \\ (Strong current)}} 
& APEN $\downarrow$ & 0.2733            & 0.2983        & 0.1515        & \textbf{0.0719} \\
& MPEN $\downarrow$ & \textbf{0.3215}   & 0.4199        & 1.0402        & \textbf{0.3369} \\
& ACE $\downarrow$  & \textbf{30.41}    & 38.68         & \textbf{30.75} & \textbf{30.45} \\
& MIR $\downarrow$  & 75.5              & 953.1         & 246.5         & \textbf{63.7} \\

\bottomrule
\end{tabular}
\end{table}

%%%%%%%%%%%%%%%%%%%%%%%%%%%%%%%%%%%%%%%%%%%%%%%%%%%%%%%%%%%%%%%%%%%%%%%%%%%%%%%%%%%

\section{Experimental Validation}
\label{sec:experiment}

This section experimentally evaluates IC-ADRC with active saturation and
low-pass filtering. The objective is to verify the simulated performance
on a physical underwater vehicle subject to actuator limitations and
environmental disturbances.

Experiments are conducted in a
$6~\mathrm{m}\times1.8~\mathrm{m}\times1~\mathrm{m}$ indoor tank using
the BlueROV2 Heavy Configuration\footnote{\url{https://bluerobotics.com/store/rov/bluerov2/}},
as shown in Fig.~\ref{fig:experiment}. The controllers run on a tethered
external computer and transmit commands via
MAVLink\footnote{\url{https://mavlink.io/en/}} at
10~Hz because of onboard bandwidth and safety constraints. Ground-truth
pose is measured using a four-camera Qualisys underwater motion-capture
system (Arqus, Qualisys AB, Gothenburg,
Sweden)\footnote{\url{https://www.qualisys.com/cameras/underwater/}}. 
Environmental disturbances are generated by a programmable wave generator
driven by sinusoidal signals with frequency $f_w$ and amplitude $A_w$
\cite{wavegen}. Disturbance rejection is evaluated at the desired
equilibrium under three conditions:

\begin{itemize}
    \item \textbf{Case 1 (Still water):} No external wave disturbance;
    only hydrodynamic uncertainties and tank-wall reflections are present.

    \item \textbf{Case 2 (Moderate wave):}
    $f_w=1~\mathrm{Hz}$ and $A_w=1~\mathrm{V}_{pp}$.

    \item \textbf{Case 3 (Large wave):}
    $f_w=0.5~\mathrm{Hz}$ and $A_w=1.5~\mathrm{V}_{pp}$.
\end{itemize}
\begin{figure}[htb!]
    \centering    
    \includegraphics[width=0.6\linewidth]{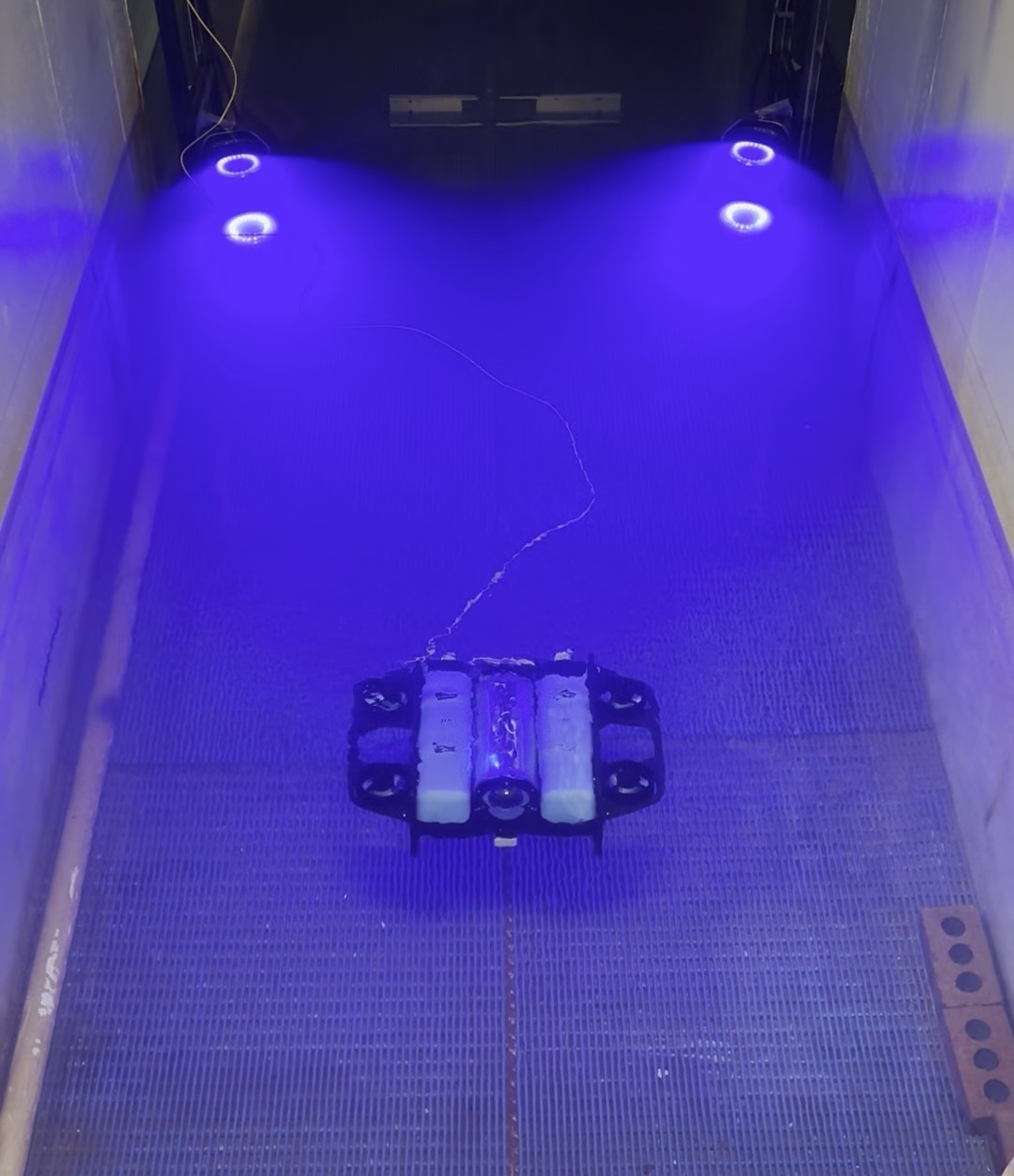}
\caption{The indoor water tank with the BlueROV2 Heavy Configuration and Qualisys Motion Capture system for the practical experiments.}
\label{fig:experiment}
\end{figure}
\begin{figure}[htb!]
    \centering
    \includegraphics[width=0.8\linewidth]{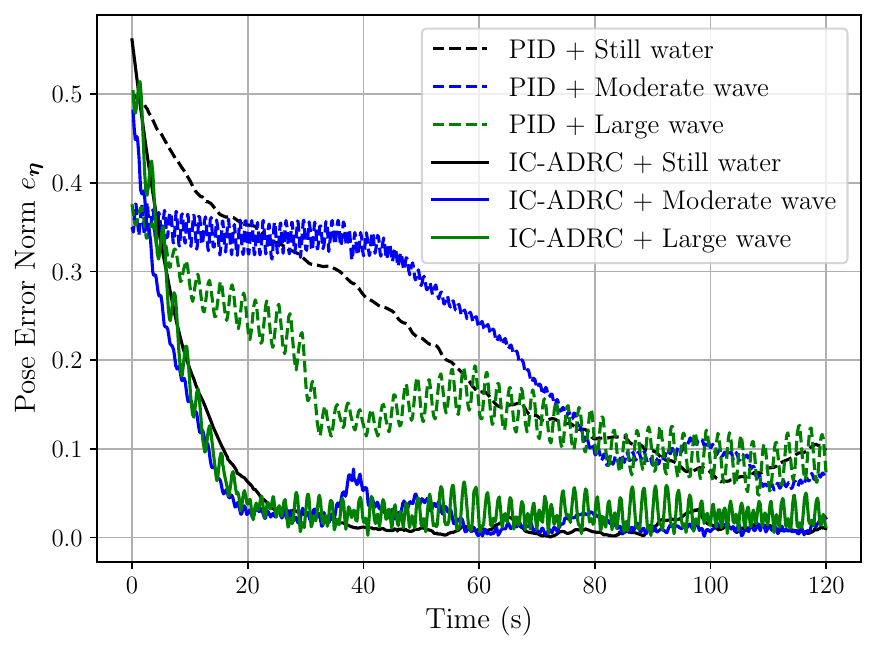}
    \caption{Real robot experiment: Pose Error Norm of {IC-ADRC} and PID in different test cases}
    \label{fig:real}
\end{figure}

IC-ADRC is compared with a benchmark PID controller representing the
industrial standard for marine platforms. Since advanced controllers were
evaluated in Section~\ref{sec:compare}, the experiments focus on this
industrial baseline under practical operating conditions. Both controllers
use the parameters from Section~\ref{sec:compare}, with the nominal
IC-ADRC parameters listed in Table~\ref{tab:nominal_params}.

Fig.~\ref{fig:real} compares the motion-capture pose errors of PID and
IC-ADRC under the three experimental conditions. Small differences in
initial position are unavoidable between runs. In all cases, IC-ADRC
achieves faster convergence and smaller errors. In still water, IC-ADRC
rapidly regulates the vehicle, whereas PID exhibits a steady-state error of
approximately $10~\mathrm{cm}$, primarily because of the thruster dead
zone.

The performance gap increases under stronger waves. PID exhibits larger
oscillations and persistent errors, whereas IC-ADRC compensates for the
unmodeled actuator dynamics and wave disturbances. These results agree
with the simulations and confirm the practical effectiveness of the
proposed time-scale-separated architecture.

%%%%%%%%%%%%%%%%%%%%%%%%%%%%%%%%%%%%%%%%%%%%%%%%%%%%%%%%%%%%%%%%%%%%%%%%%%%%%%%%%%%

\section{Conclusion}
\label{sec:con}

This paper presented an input-constrained active disturbance rejection
control (IC-ADRC) framework that combines active saturation with
post-saturation low-pass filtering to enforce AUV actuator constraints.
Singular perturbation analysis established the required time-scale
separation and semi-global practical stability. Simulations and experiments
demonstrated reduced noise-induced chattering, compliance with thruster
rate limits, and robust stabilization under environmental disturbances.
Future work will improve transient performance under rapidly varying wave
disturbances.

%%%%%%%%%%%%%%%%%%%%%%%%%%%%%%%%%%%%%%%%%%%%%%
\section*{Acknowledgment}
This work was supported by the Melbourne Research Scholarship and the
Australian Research Council under Grant DE220101527.

\bibliographystyle{ieeetr}
\bibliography{ref-short}

\begin{IEEEbiography}[{\includegraphics[width=1in,height=1.25in,clip,keepaspectratio]{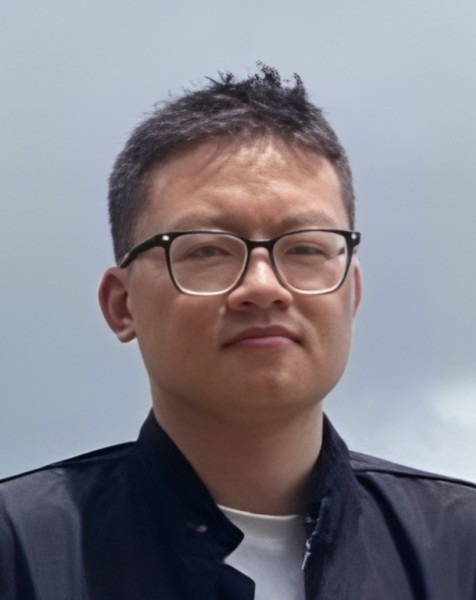}}]{Minh-Quan Le}
received his bachelor's degree in Mechatronics Engineering from Hanoi University of Science and Technology in 2023. 
Since 2024, he has been working toward a Ph.D. degree in Engineering and IT at the University of Melbourne, VIC, Australia. 
His research interests include robust control and navigation, focusing on interdisciplinary applications in underwater autonomy utilizing vision-based perception.
\end{IEEEbiography}

\begin{IEEEbiography}[{\includegraphics[width=1in,height=1.25in,clip,keepaspectratio]{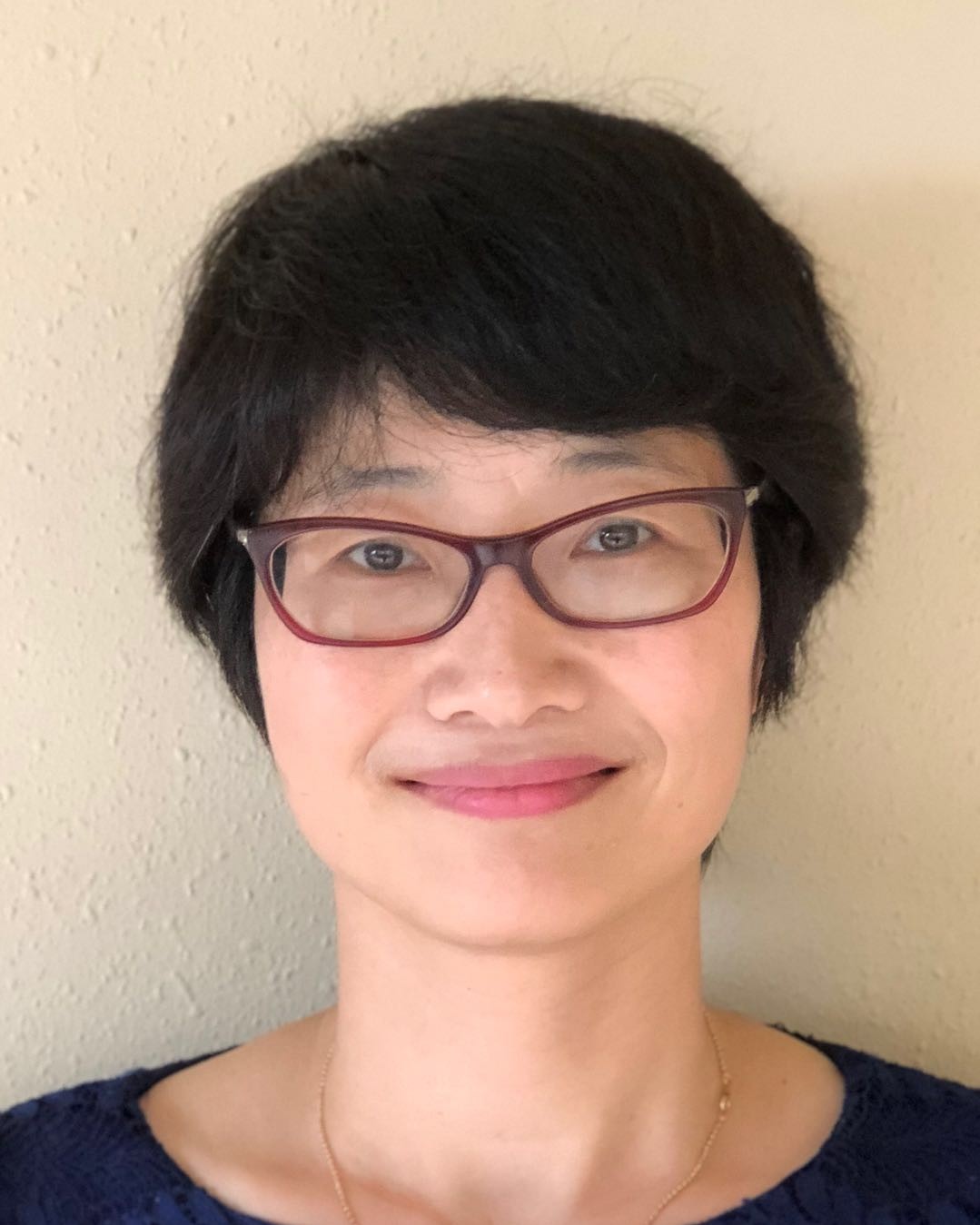}}]{Ying Tan}
(Fellow, IEEE) received the bachelor’s degree in electrical engineering from Tianjin University, Tianjin, China, in 1995, and the Ph.D. degree in electrical engineering from the National University of Singapore, Singapore, in 2002.
Following a postdoctoral fellowship with McMaster University, she joined the University of Melbourne, Melbourne, Australia, in 2004. 
She is a Professor of Mechanical Engineering at the University of Melbourne. 
Her research interests include intelligent systems, nonlinear systems, data-driven optimization, rehabilitation robotics, human motor learning, wearable sensors, and model-guided machine learning.
Dr. Tan has received numerous prestigious recognitions, including an Australian Postdoctoral Fellowship (2006–2008) and an ARC Future Fellowship (2009–2013). 
She currently serves on the ARC College of Experts (2024–2026) and holds several distinguished titles, including Fellow of Engineers Australia (FIEAust), and Fellow of the Asia-Pacific Artificial Intelligence Association. 
She is also a Member of the IEEE Fellow Committee (2024–2025) and was recently elected Fellow of the Australian Academy of Technology and Engineering (ATSE, 2025), an honor recognizing her outstanding contributions to engineering and applied science in Australia.
\end{IEEEbiography}

\begin{IEEEbiography}[{\includegraphics[width=1in,height=1.25in,clip,keepaspectratio]{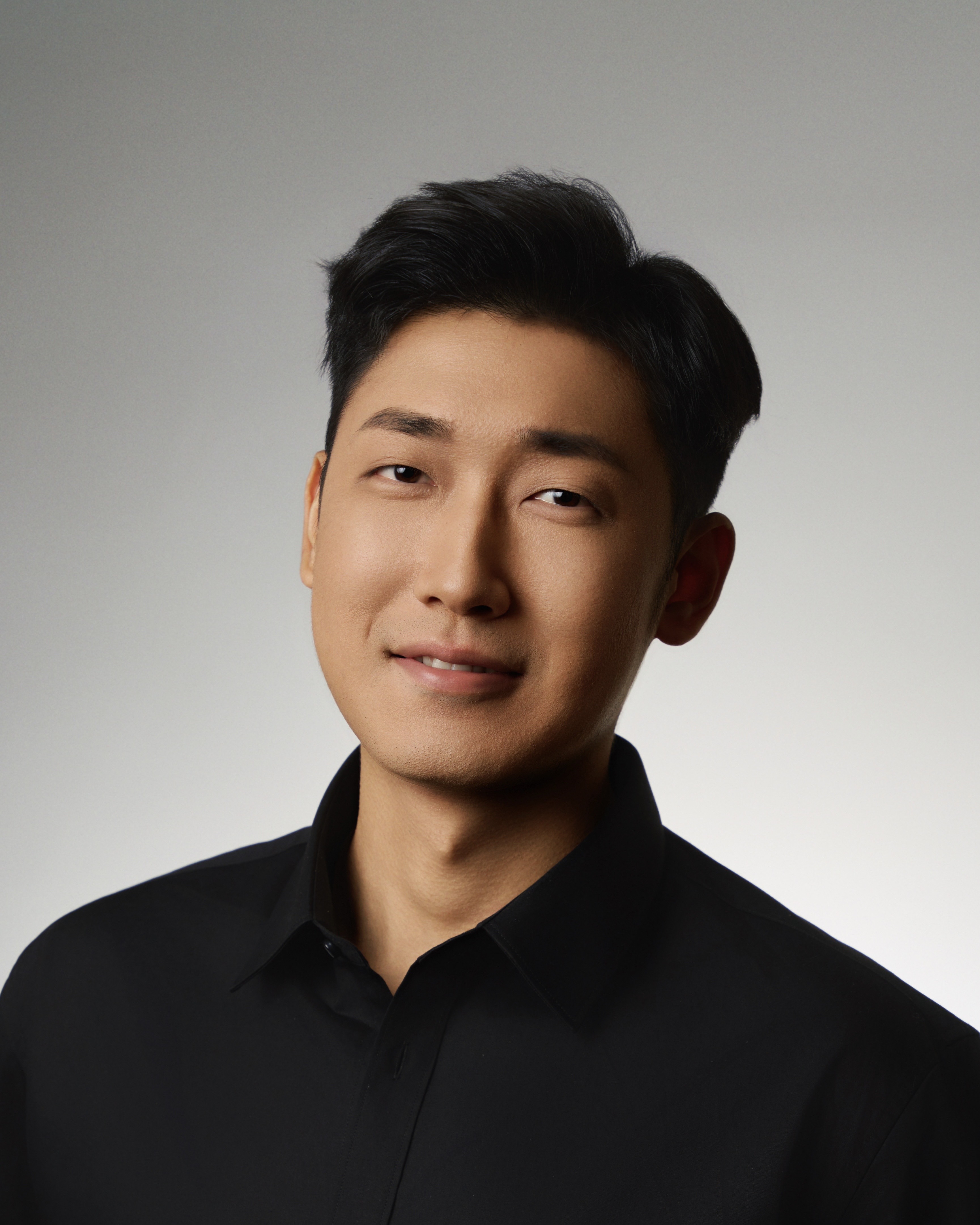}}]{Jinghe Yang}
received the Ph.D. degree in electrical and electronic engineering from the University of Melbourne, Melbourne, VIC, Australia, in 2026.
His research interests include computer vision and underwater robotics, with a particular focus on vision-based perception and navigation for autonomous underwater vehicle systems in complex environments. 
\end{IEEEbiography}

\begin{IEEEbiography}[{\includegraphics[width=1in,height=1.25in,clip,keepaspectratio]{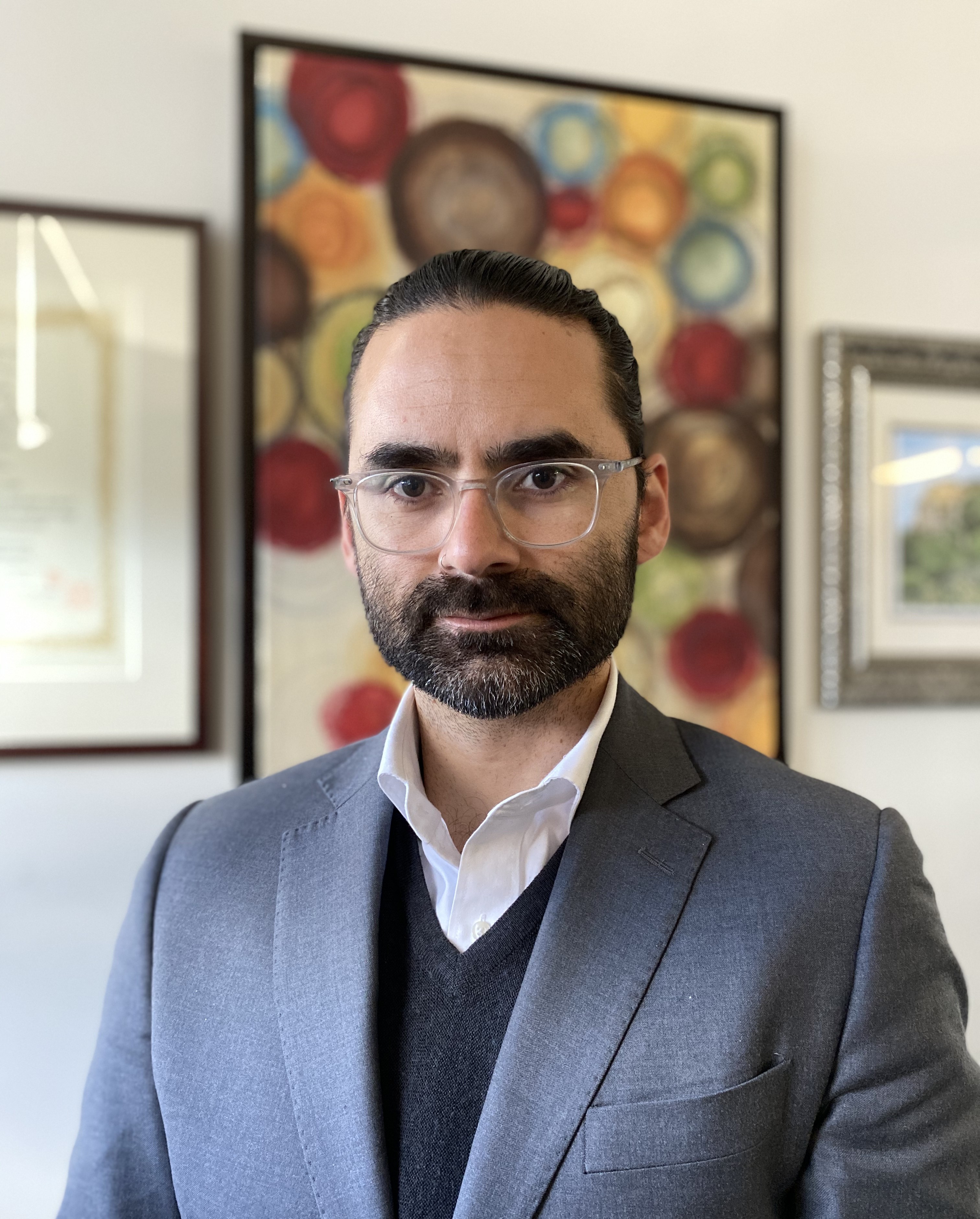}}]{Jason Monty}
is Head of School of Electrical, Mechanical and Infrastructure Engineering and a Professor with the Fluid Mechanics Group at The University of Melbourne. 
Dr Monty is a world-renowned experimental researcher in fluid mechanics who specialises in turbulent flows, including air-sea interactions, drag reduction, ship performance and ice-wave interactions. 
More recently he has turned his attention to indoor air quality, including policy and demonstration work. 
He has published over 150 peer-reviewed articles with 10,000 citations and an H-index of 49 (google scholar). 
Dr Monty is also director of the Michell Hydrodynamics Laboratory, where he has built a number of world-class experimental facilities including the 60 m long Extreme Air-Sea Interaction (EASI) facility that can produce hurricane wind-speeds. 
His research group has 1 Postdoctoral Research Fellow, 10 PhD students and 4 Research Engineers. 
\end{IEEEbiography}

\begin{IEEEbiography}[{\includegraphics[width=1in,height=1.25in,clip,keepaspectratio]{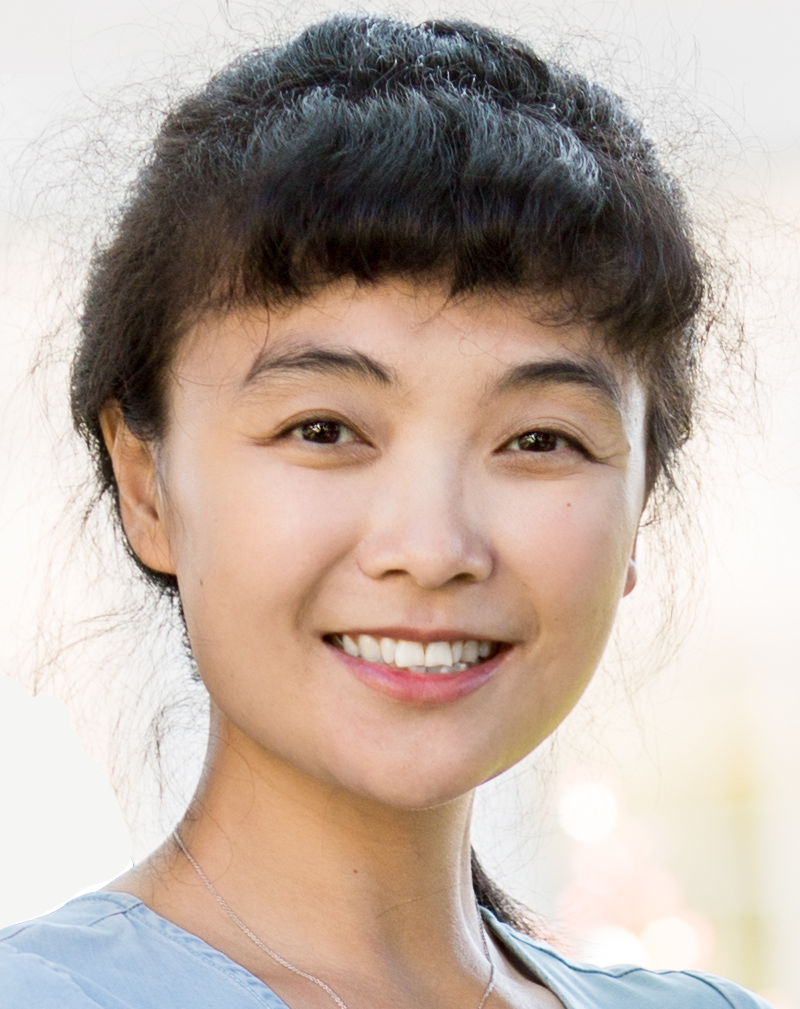}}]{Ye Pu}
received the B.S. degree from Shanghai Jiao Tong University, China, in 2008, the M.S. degree from the Technical University of Berlin, Germany, in 2011, and the Ph.D. degree from the Swiss Federal Institute of Technology Lausanne (EPFL), Switzerland, in 2016. 
From 2016 to 2018, she was a postdoctoral researcher with the University of California, Berkeley, USA. 
She is currently a Senior Lecturer with the Department of Electrical and Electronic Engineering at the University of Melbourne, Australia. 
She was awarded the Swiss National Science Foundation Early Postdoc Mobility Fellowship from 2016 to 2018 and is the recipient of the Australian Research Council (ARC) Discovery Early Career Researcher Award (DECRA) for 2022–2025.
\end{IEEEbiography}

\end{document}